\documentclass{article}

\PassOptionsToPackage{numbers, compress}{natbib}
\usepackage[final,main]{neurips_2026}

\usepackage[utf8]{inputenc}
\usepackage[T1]{fontenc}
\usepackage{hyperref}
\hypersetup{hidelinks}
\usepackage{url}
\usepackage{microtype}
\usepackage{xcolor}
\usepackage{graphicx}
\usepackage{subcaption}
\usepackage{booktabs}
\usepackage{float}
\usepackage{svg}
\usepackage{tabularray}
\usepackage{tabularx}
\usepackage{longtable}
\usepackage{tipa}
\usepackage[normalem]{ulem}
\usepackage{dirtytalk}
\usepackage{amsmath}
\usepackage{amssymb}
\usepackage{amsfonts}
\usepackage{mathtools}
\usepackage{amsthm}
\usepackage{multirow}
\usepackage{algorithm}
\usepackage{algorithmic}
\usepackage[capitalize,noabbrev]{cleveref}
\usepackage[disable,textsize=tiny]{todonotes}
\usepackage{wrapfig}

\theoremstyle{plain}
\newtheorem{theorem}{Theorem}[section]

\theoremstyle{definition}

\newtheorem{assumption}[theorem]{Assumption}
\theoremstyle{remark}

\title{HuPER: A Human-Inspired Framework\\for Phonetic Perception}

\author{
  \makebox[0pt][c]{%
    \parbox{1.15\textwidth}{%
      \centering
      \normalsize
      \textbf{Chenxu Guo$^{1,2,\ast,\dagger}$, Jiachen Lian$^{2,\ast}$, Yisi Liu$^{2}$,
      Baihe Huang$^{2}$, Shriyaa Narayanan$^{2}$, Bixing Wu$^{1}$, Zoe Ezzes$^{3}$,} \\
      \textbf{Jet Vonk$^{3}$, Zachary Miller$^{3}$, Cheol Jun Cho$^{2}$,
      Maria Gorno-Tempini$^{3}$, Gopala Anumanchipalli$^{2}$} \\[0.35em]
      \normalfont\normalsize
      $^{1}$Zhejiang University \quad
      $^{2}$University of California, Berkeley \quad
      $^{3}$University of California, San Francisco \\
      \texttt{\{jiachenlian,gopala\}@berkeley.edu}
    }%
  }
}

\begin{document}

\maketitle

\begingroup
\renewcommand{\thefootnote}{\fnsymbol{footnote}}
\footnotetext[1]{Equal contribution.}
\footnotetext[2]{Chenxu Guo was a Visiting Student Researcher at UC Berkeley during this work.}
\endgroup


\begin{abstract}
We propose \textit{HuPER}, a human-inspired framework that models phonetic
perception as adaptive inference over acoustic--phonetic evidence and linguistic
knowledge. HuPER first learns an acoustically grounded phone recognizer from limited
human-annotated data and transcript-only speech: canonical G2P transcriptions
are treated as auxiliary linguistic cues rather than ground truth, and are
corrected into realized phone proxies for self-training. At inference time,
HuPER supports multiple perceptual routes: it can rely on bottom-up phone
evidence when the signal is clear, incorporate explicit expectations when a
reference is available, or use lexical constraints when acoustic evidence is
weak. With only 100 hours of training data, HuPER achieves state-of-the-art
phonetic feature error rates on five English benchmarks and strong zero-shot
transfer to 94 unseen languages. HuPER also improves robustness on
weak-evidence and disordered speech, demonstrating the benefit of adaptive,
multi-path phonetic perception under diverse acoustic and task conditions. All
training data, models, and code are open-sourced. Code and demo are available at
\url{https://github.com/Berkeley-Speech-Group/HuPER}.
\end{abstract}

\section{Introduction}

Phonetic modeling is fundamental to speech perception. Early speech perception systems were expert systems~\cite{davis1952automatic,erman1980hearsay,lowerre1976harpy,baker1975dragon,rabiner2002tutorial} designed to emulate human perception, but constrained in scalability. Recent advances have demonstrated that large-scale word-based ASR~\cite{radford2023robust, zhang2023google-usm, pratap2024scaling-mms, omnilingual2025omnilingual} can match or surpass human parity across domains. However, progress at the phonetic level remains limited, with little comparable gains despite similar scaling efforts~\cite{li2025powsm}. These observations motivate a human-centered rethinking of phonetic foundation models.



\begin{wrapfigure}{r}{0.38\linewidth}
    \centering
    \includegraphics[width=1.0\linewidth]{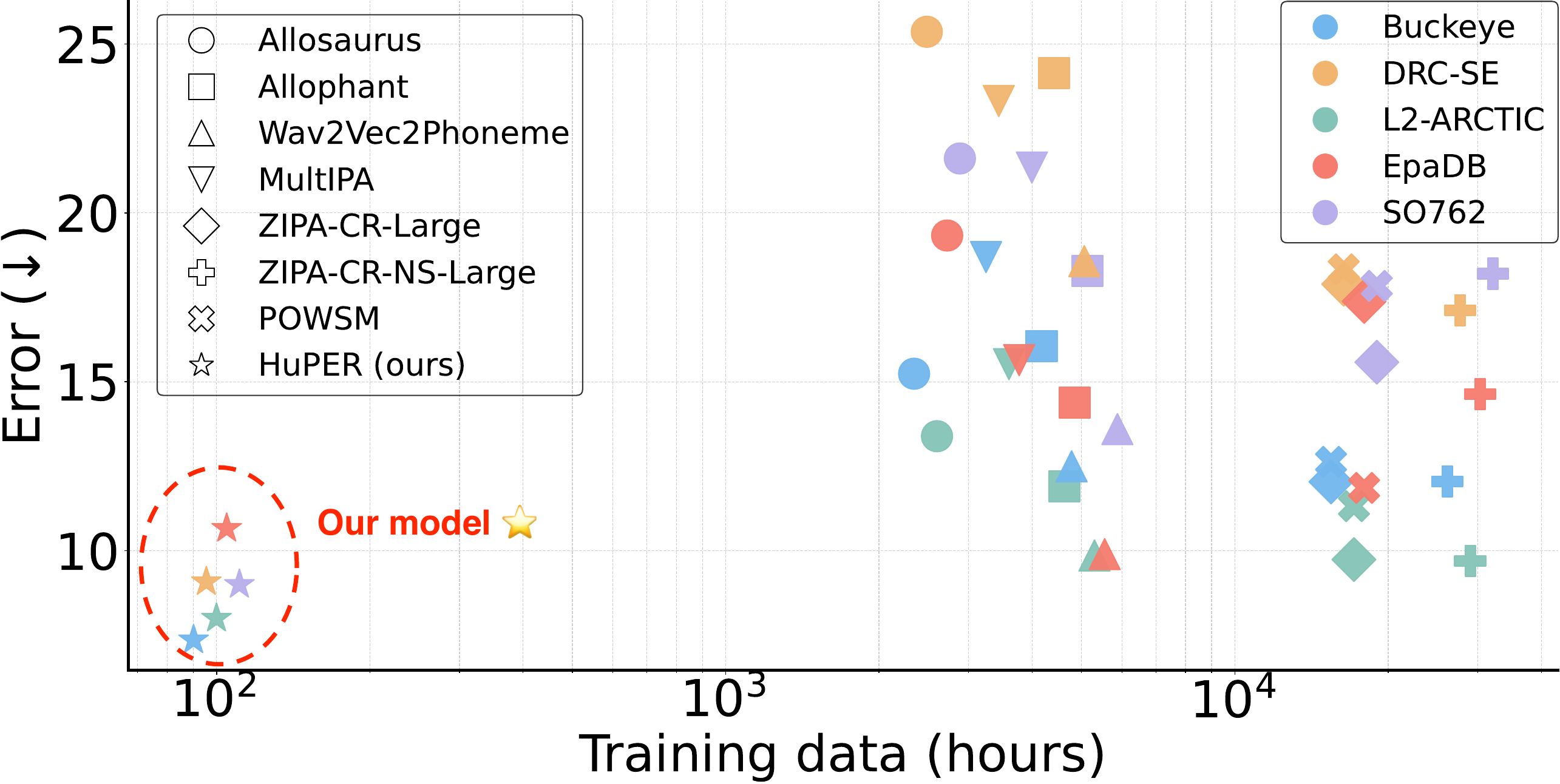}
    \caption{HuPER achieves highly data-efficient phonetic transcription on English variation benchmarks.}
    \label{fig:placeholder}
\end{wrapfigure}

\begin{figure}[t]
    \centering
    \includegraphics[width=1.0\linewidth]{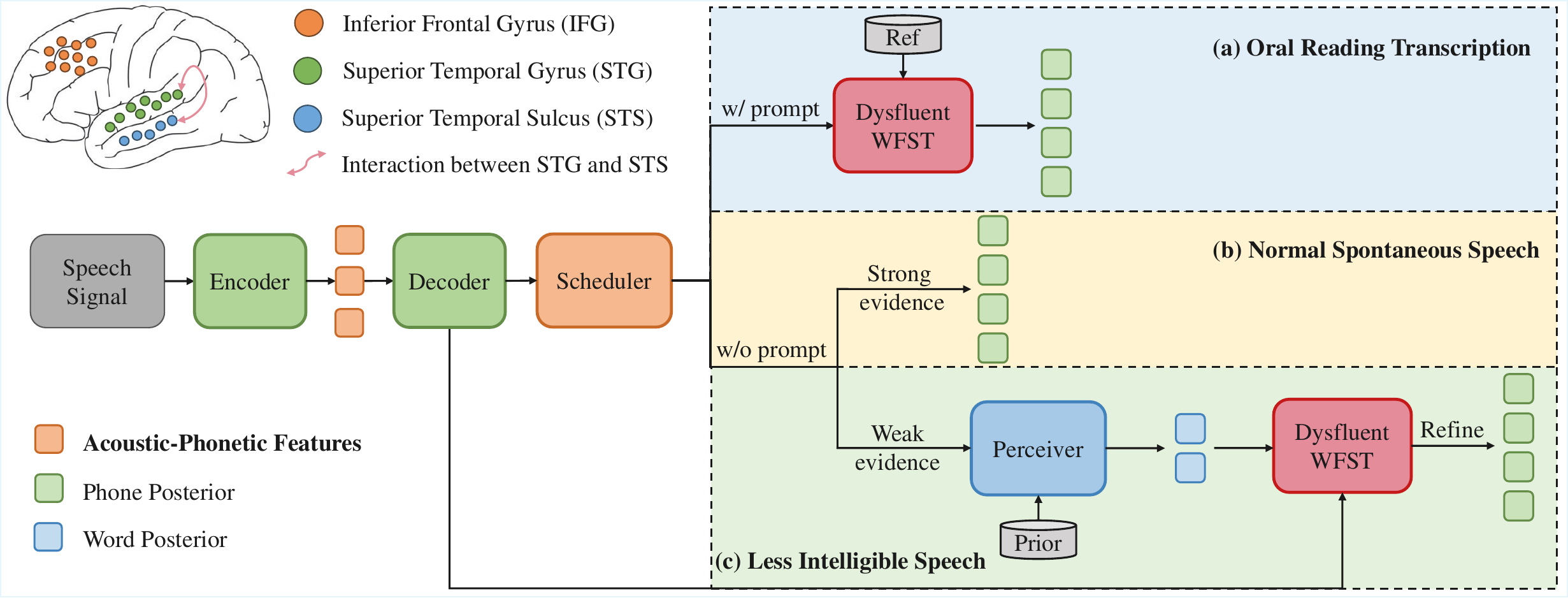}
    \caption{\textbf{HuPER overview: evidence-controlled multi-path speech perception.}
    Left: HuPER-Recognizer (Encoder/Decoder; mapped to STG) converts the speech signal into acoustic--phonetic features (orange) and phone posteriors (green). HuPER-Scheduler (mapped to IFG) monitors evidence strength and selects an inference route. Right: (a) \textit{Oral reading transcription}: with an external prompt/reference, a Dysfluent WFST applies explicit top-down constraints to produce the transcript. (b) \textit{Normal spontaneous speech}: when evidence is strong, the system trusts bottom-up phone evidence and outputs directly. (c) \textit{Less intelligible speech}: when evidence is weak, HuPER-Perceiver combines phone evidence with a lexical prior to form word hypotheses (blue), which are then refined by the Dysfluent WFST~\cite{conf/interspeech/GuoLZZLYPDEVMBW25}.}
    \label{fig:huper_human}
\end{figure}

The human phonetic processor extracts acoustic–phonetic cues~\cite{stevens2000acoustic,mesgarani2014phonetic-science} and organizes them into hierarchical phonological representations for higher-level linguistic inference~\cite{johnson2011acoustic}. This process comprises a language-universal and an  experience-dependent layer~\cite{werker1988cross} encoding language-specific structure~\cite{bhaya2026shared-nature}. Both functions are dynamically shaped by bottom-up sensory evidence and top-down perceptual processes~\cite{mesgarani2014phonetic-science, bhaya2026shared-nature}.

Early self-supervised speech learning (S3L) models~\cite{mohamed2022self-s3l} can be viewed as analogous to infants’ early exposure to continuous speech~\cite{lavechin23_interspeech-baby-slm}, during which broad acoustic–phonetic representations~\cite{choi24b_interspeech} are initially formed. However, subsequent fine-tuning~\cite{conf/interspeech/XuBA22, Chen2021WavLM} typically introduces lexical supervision \textit{prematurely}. For instance, in fast speech, “last Sunday” is often realized as “las[] Sunday.” A developmentally plausible approach would supervise such surface forms first, allowing acoustic–phonetic representations to stabilize before higher-level abstraction. In contrast, most current phonetic models rely on G2P-derived canonical targets~\cite{rudnicky1993cmu,mortensen2018epitran,black2001flite} (e.g., “last Sunday”), which encode phonological and grammatical regularities absent from the acoustic signal. This supervision mismatch prevents models from fully exploiting their acoustic–phonetic capacity and undermines subsequent phonological and lexical inference. 


Beyond the premature introduction of lexical supervision, another open challenge is that human phonetic perception operates as a dynamic closed-loop system rather than a one-pass mapping from speech to phonemes~\cite{mcclelland1986trace,hickok2007cortical,norris2003perceptual}. Most phoneme recognition models implicitly assume a unidirectional, feedforward processing pipeline. However, under ambiguous or degraded conditions, human listeners often engage top-down inference, using lexical or contextual expectations to constrain phonetic interpretation~\cite{warren1970perceptual,ganong1980phonetic,samuel2001knowing}, which in turn refines higher-level representations~\cite{mcclelland1986trace,norris2003perceptual}. Notably, even when top-down information is available, listeners can deliberately rely on bottom-up processing when instructed or required by task demands~\cite{cutler2012native}, indicating flexible cognitive control~\cite{posner1989attention,norman1986attention}. Moreover, when reference text is available, as in performed or rehearsed speech, perception follows another pathway guided by explicit expectations~\cite{lian2024ssdm}. In contrast, most current AI systems implement a single feedforward route and lack mechanisms for modeling such multi-path, closed-loop dynamics.

Given these limitations, we propose \textbf{HuPER (A Human-Inspired Framework for Phonetic Perception)}, a human-centered framework that models phonetic perception as adaptive inference integrating acoustic evidence and linguistic knowledge. HuPER departs from conventional end-to-end pipelines by explicitly coordinating bottom-up cues and top-down expectations for robust, context-aware perception. Our main contributions are summarized as follows:

(1) We present HuPER, the first unified and explicit computational framework for modeling human phonetic perception, which also provides a diagnostic perspective on existing phoneme recognition models.

(2) HuPER achieves state-of-the-art phonetic accuracy with only 100 hours of training data, reaching an average PFER of 8.82\% on English benchmarks and demonstrating decent zero-shot multilingual transfer.

(3) We propose an adaptive multi-path inference mechanism for phonetic perception, enabling dynamic pathway selection and improved robustness under degraded and disordered speech conditions.




\begin{figure}[t]
    \centering
    \includegraphics[width=1.0\linewidth]{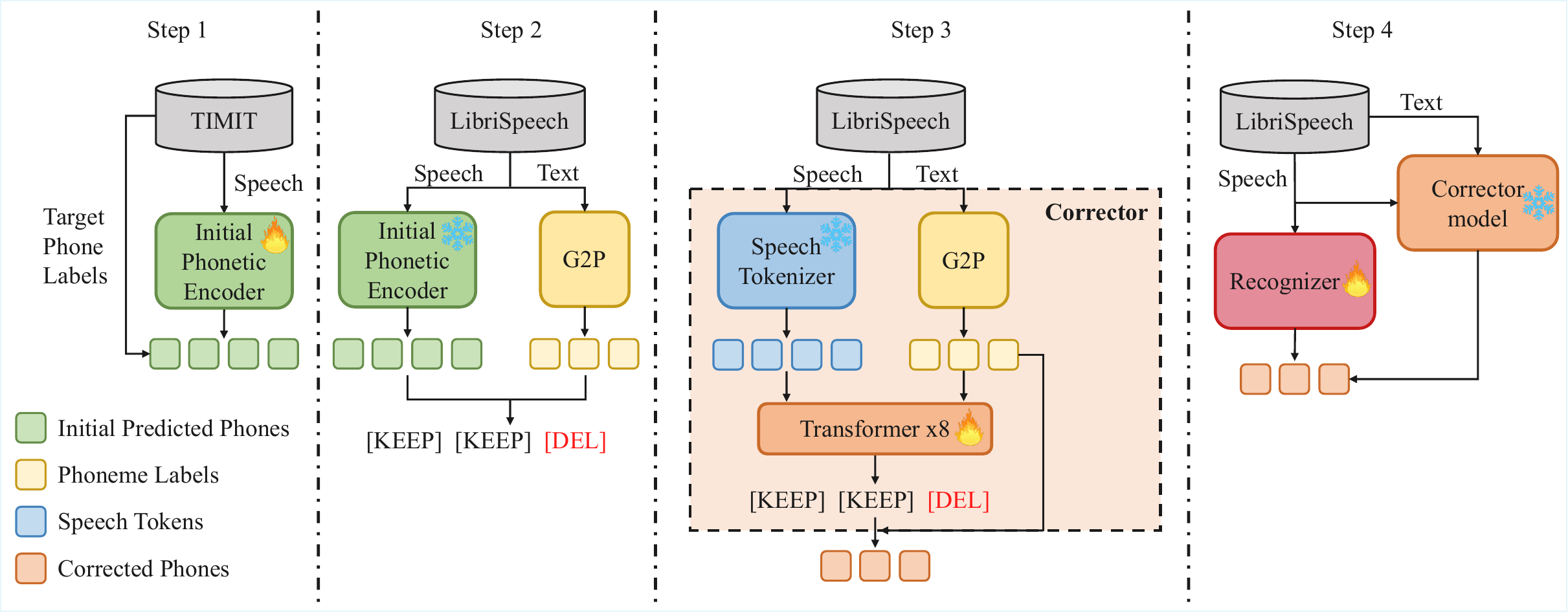}
    \caption{\textbf{HuPER-Recognizer self-learning pipeline}. HuPER uses an initial recognizer and G2P transcriptions to train a Corrector, which transforms canonical phoneme sequences into acoustically grounded phone proxies for retraining on LibriSpeech.}
    \label{fig:selflearning}
\end{figure}

\section{HuPER overview}

We propose \textbf{HuPER (A Human-Inspired Framework for Phonetic Perception)}, a modular speech perception framework that supports adaptive multi-path inference by integrating bottom-up acoustic--phonetic modeling with explicit top-down constraints, serving as an approximation to human cognitive speech perception systems~\cite{mcclelland1986trace,norris2003perceptual}. As shown in Figure~\ref{fig:huper_human}, HuPER consists of three functional modules and a central scheduler:

\textbf{(1) HuPER-Recognizer} (STG-like, Bottom-Up Acoustic--Phonetic Perception, Sec.~\ref{sec:recognizer}).  
It extracts language-general acoustic--phonetic evidence and outputs spoken phones, providing the foundational representation for downstream modules and determining the upper bound of universal phonetic perception. We further develop a self-training procedure for HuPER-Recognizer, which is theoretically grounded in our DRRC framework.
\textbf{(2) HuPER-Perceiver} (STS-like, Phonetic--Lexical Integration, Sec.~\ref{sec:perceiver}). It combines acoustic--phonetic representations with explicit lexical and phonotactic priors to generate phonetic-enhanced word hypotheses.
\textbf{(3) Dysfluent WFST} (explicit top-down constraints;~\cite{conf/interspeech/GuoLZZLYPDEVMBW25}).
It provides a human-inspired, top-down constraint mechanism by representing dysfluencies and pronunciation variants in a WFST constraint graph, which can be conditioned on external references or on hypotheses.
\textbf{(4) HuPER-Scheduler} (IFG-like, Sec.~\ref{sec:scheduler}). The scheduler selects inference pathways based on signal quality and task context. For clear speech, HuPER relies on bottom-up inference. Under degraded or ambiguous conditions, it integrates HuPER-Perceiver outputs with Dysfluent WFST constraints. In reference-guided scenarios~\cite{lian2024ssdm}, known intended text is incorporated through the constraint graph. We detail each part in the following sections.

\section{HuPER-Recognizer}
\label{sec:recognizer}

\paragraph{Task definition and outputs.}
The HuPER-Recognizer is a WavLM-Large model~\cite{Chen2021WavLM} fine-tuned for phone recognition.
Given a speech signal $X$, it outputs phone posteriors and a decoded phone sequence $\hat Y_\theta$.
We treat $\hat Y_\theta$ as the \emph{spoken phone} (i.e., what a human would perceive from the acoustics), and we aim to learn phone evidence that generalizes across languages.

To scale beyond scarce human phone labels, a common approach is self-training: generate pseudo phone labels on a large transcript-only set and retrain the recognizer on them~\citep{lee2013pseudo,xie2020noisy}.
However, naive pseudo-label training can amplify systematic errors (confirmation bias) and can be statistically biased when the availability of true phone labels is not random.

We address these issues by introducing \textbf{doubly robust risk correction (DRRC)}.
Concretely, we (i) construct proxy phone supervision on $\mathcal{D}'$ via a phoneme$\rightarrow$phone Corrector (Algorithm~\ref{alg:huper_self_learning}), and (ii) view the resulting training problem as missing-label learning and analyze it through a DRRC objective (Sec.~\ref{sec:recognizer_drrc}).

\paragraph{HuPER self-learning strategy.}
Figure~\ref{fig:selflearning} summarizes the recipe for constructing proxy phone supervision on a large transcript-only dataset.
We use a small labeled set $\mathcal{D}=\{(X_i,Y_i)\}$ with human-verified phones and a transcript-only set $\mathcal{D}'=\{(X_j,T_j)\}$.
For each $(X,T)\in\mathcal{D}'$, we compute a canonical phoneme sequence $Z=\mathrm{G2P}(T)$, obtain a teacher phone hypothesis $\bar Y$ from the current recognizer, and apply a Corrector to produce a corrected proxy phone sequence $\tilde Y$.
Algorithm~\ref{alg:huper_self_learning} summarizes the iterative training procedure.


\paragraph{DRRC perspective.}
\label{sec:recognizer_drrc}
To formalize our analysis, let
$W=(X,Z,Y,\hat Y,A)$ denote one training example in the missing-label view:
$X$ is the speech input, $Z=\mathrm{G2P}(T)$ is the canonical phoneme sequence,
$Y$ is the latent human-verified phone label, $\hat Y$ is the always-observed proxy phone label, and
$A\in\{0,1\}$ indicates whether $Y$ is observed.
Let $e(k)$ be the one-hot vector for phone class $k$; thus $e(Y)$ and $e(\hat Y)$ are the one-hot targets for the true and proxy labels.
For a fixed HuPER-Recognizer parameter $\theta$, we define the true population risk as
\begin{equation}
    R(\theta) := \mathbb{E}\big[\ell_\theta(e(Y),X)\big],
\end{equation}
Treating the canonical phoneme sequence $Z=\mathrm{G2P}(T)$ as an auxiliary covariate, we cast self-training as a \emph{missing data problem}~\citep{robins1994estimation,bang2005doubly} in which the observability of the true label $Y$ is governed by $Z$ via the true propensity function:
\begin{equation}
g^*(z,\hat y) := \mathbb{P}(A=1\mid Z=z,\hat Y=\hat y).
\end{equation}
Practically, this formulation assumes that the transcript contains sufficient auxiliary information, distinct from the features captured by the teacher model $f_T$, to distinguish the pseudo-label $\hat Y$ from the true label $Y$ (but not necessarily to fully recover the true label $Y$).
We demonstrate that there exists a corrector such that the HuPER loss yields a \emph{doubly robust} estimate of $R(\theta)$ on the self-training dataset. Specifically, the estimator is consistent if either (i) the propensity model for missing phone labels is correctly specified, or (ii) the proxy label $\hat{Y}$ is (asymptotically) accurate.

\begin{theorem}[Informal version of \Cref{thm:dr}]
For any measurable $g:\mathcal{Z}\times\{1,\dots,K\}\to(0,1]$, define
\begin{equation}
C_g(W) := e(\hat Y) + \frac{A}{g(Z,\hat Y)}\big(e(Y)-e(\hat Y)\big).
\end{equation}
Let $\hat g$ be a cross-fitted estimator clipped to $[\varepsilon,1]$ and define
\begin{equation}
\hat R_n(\theta) := \frac{1}{n}\sum_{i=1}^n \ell_\theta\!\big(C_{\hat g}(W_i),X_i\big).
\end{equation}
Then $\mathbb{E}[\ell_\theta(C_g(W),X)] = R(\theta)$ provided either \textbf{(Prop)} $g = g^*$ or
\textbf{(Proxy)} $\hat Y = Y$. Furthermore, $\hat R_n(\theta)\to R(\theta)$ in probability provided either:
\begin{equation}
\textbf{(Prop): } \mathbb{E}[\,|\hat g(Z,\hat Y)-g^*(Z,\hat Y)|\,]\to 0,
\\ \text{or~}
\textbf{(Proxy): } \mathbb{P}(\hat Y\ne Y)\to 0.
\end{equation}
\end{theorem}


This theorem shows that the corrected target $C_g(W)$ recovers the true risk $R(\theta)$ in expectation, and the empirical estimator is consistent when either the proxy labels are accurate (\textbf{Proxy}) or the propensity model is correct (\textbf{Prop}).
Thus, DRRC is robust to misspecification of either component alone.

In summary, we scale HuPER-Recognizer with a phoneme$\rightarrow$phone Corrector (Algorithm~\ref{alg:huper_self_learning}) and analyze the resulting self-learning objective through DRRC, which is consistent under either accurate proxies or a correct propensity model.
Next, HuPER-Perceiver converts the phone evidence into word transcripts using explicit lexical and LM constraints.

\section{HuPER-Perceiver}
\label{sec:perceiver}

HuPER-Perceiver converts HuPER-Recognizer's \emph{phone evidence} into word transcripts by composing it with explicit, auditable language constraints, following the classic HMM--GMM decoding recipe (acoustic model + lexicon + language model) but with a phone recognizer as the evidence source~\citep{rabiner2002tutorial,mohri2002weighted}.
Given an utterance $X$ and recognizer parameters $\theta$, the output is a word sequence $\hat T$.

\paragraph{Acoustic evidence as a phone lattice.}
Rather than using only the 1-best phone sequence $\hat Y_\theta$, we represent the recognizer output as a weighted phone acceptor (phone lattice) $\Pi_\theta(X)$.
Each path $y$ in $\Pi_\theta(X)$ corresponds to a candidate phone sequence, and is assigned a cost
$C_\theta(y\mid X)$ derived from the negative log evidence (e.g., frame-level phone posteriors or arc weights).
The decoded phone sequence $\hat Y_\theta$ is simply the best path in this lattice, while $\Pi_\theta(X)$ retains uncertainty needed for downstream constrained search.

\paragraph{Imposing lexical and linguistic constraints.}
To map phone hypotheses to words, we introduce a phone-to-word transducer $L$ (lexicon) and a word-level acceptor $G_{\mathrm{LM}}$ (language model).
$L$ restricts which phone sequences realize valid words, and $G_{\mathrm{LM}}$ provides sequence-level linguistic preferences.
Both $L$ and $G_{\mathrm{LM}}$ are modular constraints that can be swapped across domains without retraining the HuPER-Recognizer.

\paragraph{Search via WFST composition.}
Decoding is performed by composing the phone evidence with the constraints and extracting the shortest path in the unified search space:
\begin{equation}
\hat{T} \;=\; \mathrm{Output}\!\left(\mathrm{ShortestPath}\!\left(\Pi_\theta(X)\circ L \circ G_{\mathrm{LM}}\right)\right).
\end{equation}

\section{HuPER-Scheduler}
\label{sec:scheduler}

The HuPER-Scheduler acts as the system's \emph{planner}, orchestrating the flow between bottom-up phone evidence and top-down linguistic expectations.
By evaluating an \emph{evidence distortion score} $s(X)$ computed from HuPER-Recognizer emissions, it decides whether to decode phones directly, or to activate a reference-constrained refinement path.
In the guided path, we compile a \emph{Dysfluent WFST} constraint $\mathcal{H}(\cdot)$ from a reference word sequence (external $R$ or Perceiver 1-best hypothesis $\hat T$), and refine phone decoding by composing $\mathcal{H}$ with the recognizer evidence~\citep{conf/interspeech/GuoLZZLYPDEVMBW25,mohri2002weighted}.
Algorithm~\ref{alg:huper_scheduler} summarizes the routing logic.

\paragraph{Quantifying Evidence Distortion.}
To assess signal reliability, the Scheduler monitors the Recognizer's frame-level logits $\mathbf{z} \in \mathbb{R}^{N \times K}$ and posteriors $\mathbf{p}_t = \mathrm{softmax}(\mathbf{z}_t)$, where $N$ is the number of acoustic frames and $K$ is the phone-inventory size.
We quantify uncertainty through the posterior margin $m_t$ and normalized entropy $h_t$:
\begin{equation}
    m_t := p_{t,(1)} - p_{t,(2)}, \quad
    h_t := \frac{-\sum_{k=1}^K p_{t,k} \log p_{t,k}}{\log K}.
\end{equation}
Here $p_{t,(1)}$ and $p_{t,(2)}$ denote the largest and second-largest posterior probabilities at frame $t$.
These are combined into a frame-level distortion proxy
$d_t := \mathrm{clip}\!\left(\tfrac{1}{2}(1-m_t) + \tfrac{1}{2} h_t, 0, 1\right)$,
and aggregated into an utterance-level score
\begin{equation}
s(X) := \frac{1}{N}\sum_{t=1}^N d_t,
\label{eq:distortion_score}
\end{equation}
which drives routing decisions.

\paragraph{Dysfluent WFST constraint $\mathcal{H}(\cdot)$.}
Given a reference word sequence $U$ (either an external reference $R$ or a hypothesis $\hat T$), the Dysfluent WFST compiles a phone-space constraint $\mathcal{H}(U),$
which encodes a bounded set of plausible \emph{realized} phone sequences around the canonical pronunciation of $U$, while allowing dysfluent edits (insertions/deletions/substitutions)~\citep{conf/interspeech/GuoLZZLYPDEVMBW25}.
This constraint is then combined with the HuPER-Recognizer phone evidence for constrained shortest-path inference.


\section{Experimental setup}

As described in Sec.~\ref{sec:scheduler}, the HuPER-Scheduler supports three tasks:
(1) \textbf{Task-1:} Speech-only transcription with strong acoustic evidence, corresponding to standard phoneme recognition.
(2) \textbf{Task-2:} Speech-only transcription under varying signal quality, where the scheduler adaptively selects between a purely bottom-up path (for strong evidence) and a combined bottom-up–top-down path (for weak evidence).
(3) \textbf{Task-3:} Transcription with reference text provided, in which bottom-up and top-down perception are jointly performed.
For \textbf{Task-1}, we conduct \textit{phone recognition} experiments (Sec.~\ref{subsec:phone_eval}). For \textbf{Task-2 and Task-3}, we conduct \textit{multi-path speech perception} experiments (Sec.~\ref{subsec:weak_evidence}).

\subsection{Phone recognition}
\label{subsec:phone_eval}

\paragraph{Setup.}
We use WavLM-Large as the backbone and fine-tune it with a CTC~\cite{graves2006connectionist} objective for phone recognition.
Training follows our self-learning recipe.
We first train an initial HuPER-Recognizer on TIMIT~\cite{garofolo1993timit}.
We then apply this initial model to LibriSpeech~\cite{panayotov2015librispeech} to obtain teacher pseudo phone labels, and train a correction model (Sec.~\ref{sec:recognizer}) that refines the canonical G2P~\cite{mortensen2018epitran,black2001flite,conf/interspeech/ZhuZJ22} phone sequence using acoustic evidence.
The Corrector is trained for 49 epochs (about 40 minutes) on 2$\times$A6000 GPUs.
Audio is tokenized with HuBERT~\cite{9585401} units.
We use dropout $0.2$, AdamW~\cite{loshchilov2017decoupled} with $\beta=(0.9,0.999)$, and learning rate $2\times 10^{-4}$.
Finally, we train HuPER-Recognizer on the corrected pseudo labels for 100 epochs on 8$\times$A6000 GPUs (batch size 12, learning rate $3\times 10^{-5}$).
We freeze the WavLM transformer for the first 24k updates and train only the linear CTC head, then unfreeze all WavLM layers for full fine-tuning.
The final CTC training loss is $0.036$.

\paragraph{Evaluation metric: PFER.}
We report \textbf{Phonetic Feature Error Rate (PFER)}, a PanPhon-based articulatory-feature edit distance widely used in multilingual phone recognition~\cite{mortensen2016panphon}.
PFER computes a minimum-cost edit distance between hypothesis and reference phone sequences, where substitutions are weighted by distinctive-feature differences (thus giving partial credit to phonetically similar phones), and is normalized by the reference length.
Unless otherwise stated, PFER is reported on a fractional scale (not multiplied by 100); Table~\ref{tab:pfer_results} reports the same quantity as a percentage.
Formal definition and costs are provided in Appx.~\ref{app:pfer}.

\paragraph{Baselines.}
We compare against widely used open-source universal phone recognizers (Table~\ref{tab:model_checkpoints}): \textbf{Allosaurus}~\cite{li2020universal}, \textbf{Allophant}~\cite{glocker23_interspeech}, \textbf{Wav2Vec2Phoneme}~\cite{conf/interspeech/XuBA22,baevski2020wav2vec,conf/interspeech/ConneauBCMA21}, \textbf{MultIPA}~\cite{chen24c_interspeech}, \textbf{ZIPA}~\cite{conf/acl/ZhuSCM25}, and \textbf{POWSM}~\cite{li2025powsm}.

\paragraph{Evaluation datasets.}
We evaluate only on corpora with \textit{human-annotated} or \textit{human-verified} phone labels, so the error rates are grounded in human perception rather than automatic alignments.
For analysis, we group the evaluation sets into \emph{English variation} and \emph{unseen-language transfer}. Dataset statistics and evaluation splits are summarized in Table~\ref{tab:eval_datasets_phone}.

\begin{table*}[ht]
  \centering
  \small
  \setlength{\tabcolsep}{4pt}
  \renewcommand{\arraystretch}{1.05}
  \caption{\textbf{PFER (\%, $\downarrow$) on human-annotated / human-verified phone datasets.}
    All baseline numbers are reproduced using released checkpoints under a unified evaluation pipeline.
    Train (h) denotes the amount of training data reported/used by each model.
    DRC-SE = DoReCo South-England; SO762 = SpeechOcean762.
    $^\ast$Allosaurus reports training scale in utterance counts; we convert to hours assuming an average utterance duration of 4 seconds.}

  \label{tab:pfer_results}
  \begin{tabular}{l c r r r r r r}
    \toprule
    \textbf{Model} & \textbf{Train (h)} & \textbf{Buckeye} & \textbf{DRC-SE} & \textbf{L2-ARCTIC} & \textbf{EpaDB} & \textbf{SO762} & \textbf{Avg.} \\
    \midrule
    Allosaurus~\cite{li2020universal} & 2,600* & 44.03 & 25.36 & 13.03 & 12.82 & 16.73 & 22.39 \\
    Allophant~\cite{glocker23_interspeech} & 4,628 & 35.04 & 24.13 & 11.91 & 14.03 & 18.11 & 20.64 \\
    Wav2Vec2Phoneme~\cite{conf/interspeech/XuBA22} & 5,300 & 27.50 & 18.57 & 8.77 & \textbf{9.59} & 14.62 & 15.81 \\
    MultIPA~\cite{chen24c_interspeech} & 3,600 & 18.69 & 23.31 & 15.52 & 15.64 & 21.34 & 18.90 \\
    ZIPA-CR-Large~\cite{conf/acl/ZhuSCM25} & 17,132 & 31.24 & 17.89 & 9.74 & 11.75 & 15.58 & 17.24 \\
    ZIPA-CR-NS-Large~\cite{conf/acl/ZhuSCM25} & 28,983 & 31.05 & 17.12 & 8.54 & 11.63 & 18.20 & 17.31 \\
    POWSM~\cite{li2025powsm} & 17,132 & 31.63 & 18.33 & 11.32 & 11.86 & 17.84 & 18.20 \\
    \midrule
    \textbf{HuPER-Recognizer (ours)} & \textbf{100} & \textbf{7.36} & \textbf{9.08} & \textbf{8.00} & 10.66 & \textbf{9.00} & \textbf{8.82} \\
    \bottomrule
  \end{tabular}

  \vspace{0.3em}
\end{table*}

\begin{figure*}
    \centering 
    \includegraphics[width=1.0\linewidth]{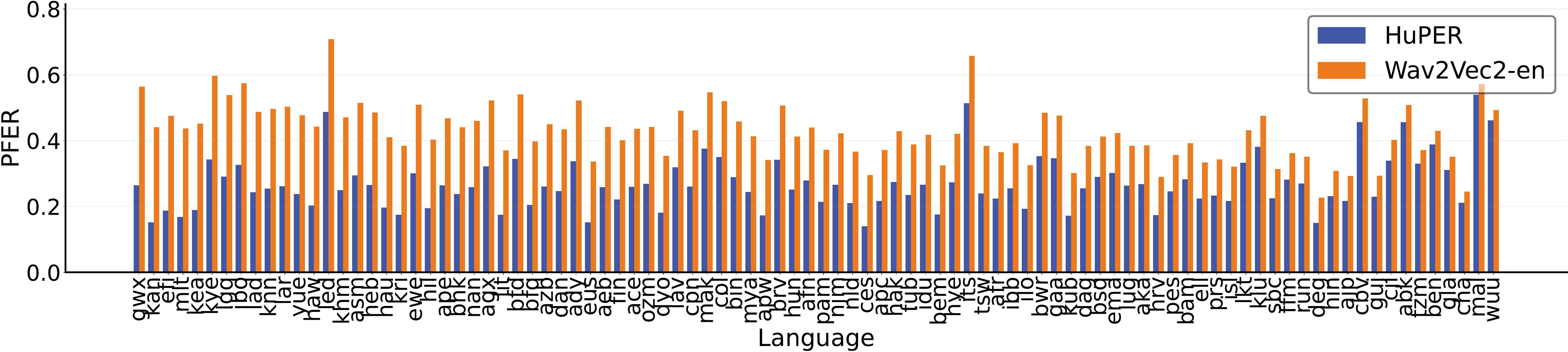} 
    \caption{\textbf{Zero-shot multilingual phone recognition on VoxAngeles (94 languages).}
    Per-language phone error rate (PFER fraction, $\downarrow$) for HuPER-Recognizer (trained only on English) and an English-only public phoneme model \texttt{Wav2Vec2-en} (fine-tuned on G2P-generated phoneme labels from LJSpeech).
    HuPER improves on the majority of languages and reduces the macro-average PFER fraction from 0.35 to 0.19.}
    \label{fig:zero_shot_multilingual}
\end{figure*}

\textbf{English variation.}
We evaluate on Buckeye~\cite{pitt2005buckeye}, which contains spontaneous conversational English.
To probe dialectal variation, we additionally test on DRC-SE (DoReCo South-England), a dialect subset from DoReCo~\cite{paschen2020doreco}.
To measure robustness to non-native pronunciations, we evaluate on L2-ARCTIC~\cite{zhao2018l2arctic}, EpaDB~\cite{vidal2019epadb}, and SpeechOcean762~\cite{zhang2021speechocean762}, which contain L2 English speech with verified phone-level annotations.
For L2-ARCTIC, we use the manually annotated \emph{perceived} transcriptions rather than dictionary/G2P pronunciations, so the reference reflects what speakers actually produced.

\textbf{Unseen languages (zero-shot).}
We further test multilingual transfer on VoxAngeles~\cite{conf/coling/ChodroffPBM24}, a post-processed version of the UCLA Phonetics Lab Archive with human-verified transcriptions spanning 94 languages.
This setting evaluates zero-shot generalization to languages and phone inventories never seen during training.

\subsection{Multi-path speech perception}
\label{subsec:weak_evidence}

\paragraph{Dataset.}
For both \textbf{Task-2} and \textbf{Task-3}, we evaluate on a primary progressive aphasia (PPA) reading dataset~\cite{gorno2011classification} (\texttt{nfvPPA}, 35 speakers).
We manually annotated \emph{spoken} phone sequences (1h1m12s total).
Each utterance provides audio $X$, a passage text $R$ (available but optionally hidden), and a human phone reference $Y$.

\paragraph{Experimental Configuration.}
For \textbf{Task-3}, we follow the Dysfluent-WFST pipeline~\cite{conf/interspeech/GuoLZZLYPDEVMBW25}, replacing the original phonetic encoder with our HuPER-Recognizer. For \textbf{Task-2} with dynamic planning, we first compute a distortion score $s(X)$ from HuPER-Recognizer posteriors (Sec.~\ref{sec:scheduler}) to characterize recognition difficulty. We then analyze whether $s(X)$ correlates with recognition performance and whether hypothesis-guided constrained decoding provides greater benefits than bottom-up 1-best decoding under high distortion. For both tasks, we treat the switching threshold $\tau$ as an externally specified hyperparameter and report results across a sweep of $\tau$. The distortion-based configuration for \textbf{Task-2} is presented as follows:

\paragraph{Distortion Computation (Task-2).}
For each utterance, we compute an utterance-level distortion score $s(X)$ by aggregating the frame-level proxy $\{d_t\}_{t=1}^T$ defined in Sec.~\ref{sec:scheduler}. We compare two phonetic encoders (HuPER and Wav2Vec2Phoneme~\cite{conf/interspeech/XuBA22}) under three decoding modes:
(i) \textbf{1-best}: CTC decoding from phone evidence $\Pi(X)$;
(ii) \textbf{refine}: constrained decoding with a hypothesis-conditioned WFST $\mathcal{H}(\hat{T}(X))$~\cite{conf/interspeech/GuoLZZLYPDEVMBW25} built from Perceiver hypotheses. We report Spearman's rank correlation coefficient $\rho$ between $s(X)$ and PFER across PPA utterances to validate $s(X)$ as an uncertainty proxy (Fig.~\ref{fig:ppa_distortion_corr}).

\paragraph{Dynamic Switching Analysis (Task-2).}
To quantify when constrained inference is beneficial under dynamic planning, we define a distortion-threshold switch between HuPER 1-best and HuPER refine decoding:
\begin{equation}
    \hat{Y}_{\tau}(X)=
    \begin{cases}
    \hat{Y}_{\mathrm{1best}}(X), & s(X)\le\tau,\\
    \hat{Y}_{\mathrm{refine}}(X), & s(X)>\tau,
    \end{cases}
\end{equation}
We report the average PFER as a function of $\tau$ to identify distortion regimes in which constrained decoding yields the largest gains.

\section{Results} 


\subsection{English phone recognition (Task-1)}
\label{subsec:results_recognizer}

\textbf{HuPER-Recognizer achieves the best average PFER with orders-of-magnitude less supervision and a more compact label space, demonstrating the effectiveness and transferability of our DRRC-based acoustic--phonetic learning.}
Table~\ref{tab:pfer_results} reports PFER on five corpora with human-annotated (or human-verified) phone labels.
HuPER-Recognizer achieves the \textbf{best average} PFER (8.82\%) while using only 100 hours of English training data, whereas all compared baselines rely on orders-of-magnitude more supervision.
This result highlights that our DRRC recipe (Sec.~\ref{sec:recognizer}) can turn a small amount of human phone annotation into strong, transferable acoustic--phonetic evidence. We note that HuPER uses a compact phone inventory with only 42 symbols, which is smaller than the vocabularies used by several baselines.
A smaller inventory reduces representational resolution and can make certain fine-grained contrasts unexpressible, which may disadvantage HuPER under PFER (a feature-sensitive edit distance).
Therefore, the average-best performance is achieved despite a conservative label space.

\subsection{Zero-shot multilingual transfer (Task-1)}
\textbf{HuPER-Recognizer matches strong multilingual baselines in strict zero-shot cross-lingual transfer, despite being trained only on English.}
We evaluate cross-lingual generalization on VoxAngeles (94 languages) in a strict zero-shot setting.
Trained only on English, HuPER-Recognizer achieves a macro-average PFER fraction of 0.19.
In contrast, a public English-only phoneme model \texttt{Wav2Vec2-en}, fine-tuned on G2P-generated phoneme labels from LJSpeech, yields a much higher PFER fraction of 0.35.
Figure~\ref{fig:zero_shot_multilingual} reports per-language results: HuPER improves over \texttt{Wav2Vec2-en} on the majority of languages, indicating robust transfer to unseen phone inventories.
Moreover, HuPER's macro-average PFER fraction of 0.19 is strongly competitive with the best multilingual baselines, despite HuPER being trained only on English.
Full per-language results are provided in Table~\ref{tab:voxangles}.

\subsection{Flexible speech perception (Task-2, Task-3)}
\label{sec:weak_results}

\textbf{Task-2: HuPER achieves substantial gains under weak acoustic evidence through distortion-guided activation of top-down constraints.}
On \texttt{nfvPPA}, the distortion score correlates positively with HuPER 1-best PFER (Figure~\ref{fig:ppa_distortion_corr}), validating it as an uncertainty proxy. While HuPER 1-best already improves over Wav2Vec2Phoneme, distortion-controlled switching yields the main gains, substantially reducing overall PFER (Table~\ref{tab:ppa_weak}). Threshold analysis and distortion-bin breakdown (Figures~\ref{fig:ppa_tau_sweep}, \ref{fig:ppa_distortion_bins}) show that improvements concentrate in high-distortion segments, where constrained decoding is most effective. Failure cases are discussed in Appendix~\ref{sec:failure_cases}.

\textbf{Task-3: HuPER-Recognizer consistently outperforms existing phonetic encoders on disordered speech, indicating more robust acoustic--phonetic representations.} 
With reference text, constrained decoding further improves robustness, and HuPER + reference achieves the best overall PFER on \texttt{nfvPPA} (Table~\ref{tab:ppa_weak}), outperforming the Dysfluent-WFST baseline.

\begin{table}
\centering
\small
\caption{\textbf{Weak-evidence speech perception on \texttt{nfvPPA}.}
Overall PFER for different decoding modes; HuPER-switched applies refinement only when $s(X)>\tau^\star$ (refine rate shown).}
\label{tab:ppa_weak}
\vspace{4pt}
\begin{tabular}{l c c}
\toprule
Method & Overall PFER $\downarrow$ & Refine rate \\
\midrule
Wav2Vec2Phoneme 1-best        & 0.46          & --     \\
HuPER-recognizer 1-best       & 0.44          & --     \\
HuPER switched ($\tau^\star$) & \textbf{0.38} & 31.1\% \\
\midrule
Dysfluent WFST                & 0.35          & 100\%  \\
HuPER + reference             & \textbf{0.32} & 100\%  \\
\bottomrule
\end{tabular}
\end{table}

\begin{figure}[t]
    \centering
    \captionsetup[subfigure]{font=footnotesize}
    \begin{subfigure}[t]{0.32\linewidth}
        \centering
        \includegraphics[width=\linewidth]{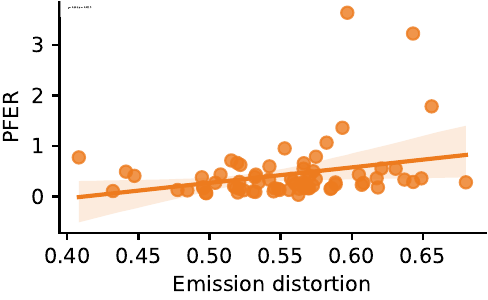}
        \caption{1-best PFER vs.\ distortion.}
        \label{fig:ppa_distortion_corr}
    \end{subfigure}\hfill
    \begin{subfigure}[t]{0.32\linewidth}
        \centering
        \includegraphics[width=\linewidth]{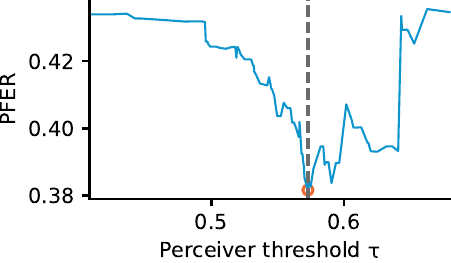}
        \caption{Threshold sweep for switching.}
        \label{fig:ppa_tau_sweep}
    \end{subfigure}\hfill
    \begin{subfigure}[t]{0.32\linewidth}
        \centering
        \includegraphics[width=\linewidth]{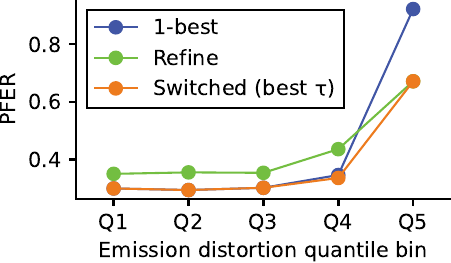}
        \caption{Refine helps in high-distortion bins.}
        \label{fig:ppa_distortion_bins}
    \end{subfigure}

    \caption{\textbf{Distortion as a control signal for refinement on PPA.}
    \textbf{(a)} Using HuPER 1-best throughout, emission distortion is positively correlated with PFER, indicating it tracks weak-evidence difficulty.
    \textbf{(b)} Sweeping the switching threshold $\tau$ reveals an optimal region for triggering refinement: we output 1-best when $s(X)\!\le\!\tau$, and otherwise invoke the perceiver to decode words from phone evidence and refine the phone sequence.
    \textbf{(c)} Stratifying utterances by distortion explains why switching helps: refinement yields the largest gains in high-distortion regimes, while low-distortion cases are best handled by direct decoding.}

    \label{fig:ppa_distortion_triptych}
\end{figure}

\section{Understanding HuPER-Recognizer Gains}

\begin{figure}[h]
    \centering
    \includegraphics[width=0.7\linewidth]{figure/feature-compare.pdf}
    \caption{(a) \textbf{Centroid RSA to acoustic--phonetic geometry.}
    Spearman correlation between pairwise phone-centroid cosine distances and PanPhon distinctive-feature distances across layers.
    Higher values indicate stronger acoustic--phonetic organization. (b) \textbf{Emission-level canonical-restoration diagnostic.}
Scatter of canonical-vs-realized preference scores on identical utterances: HuPER (x-axis) vs.\ XLSR (y-axis), colored by category. Points above the diagonal ($y=x$) indicate stronger canonical restoration by the baseline.
In panel (a), the label \emph{W2V2-eSpeak} denotes the \emph{Wav2Vec2Phoneme} checkpoint listed in Table~\ref{tab:model_checkpoints}.
}
    \label{fig:phfeat_rsa_centroid}
\end{figure}

\label{sec:analysis}


\subsection{Embedding analysis: acoustic--phonetic geometry}
\label{subsec:analysis_phgeom}

\textbf{HuPER maintains stronger alignment to distinctive-feature geometry in mid/late layers than an English G2P-supervised baseline.}
We measure layer-wise centroid RSA between phone representations and PanPhon distinctive-feature distances~\cite{mortensen2016panphon} on TIMIT.
Figure~\ref{fig:phfeat_rsa_centroid} shows that all models peak in early layers (around layer 2--4), but \textbf{HuPER remains more aligned than \emph{WavLM-Libri} deeper in the network}, suggesting DRRC-style acoustic correction helps preserve realized-phone structure beyond shallow acoustic features.
Full RSA setup is in Appendix~\ref{app:analysis_rsa}, and additional acoustic/articulatory reference analyses are in Appendix~\ref{app:mocha_zero_shot}.


\subsection{Emission analysis: canonical restoration}
\label{subsec:analysis_emission}

\textbf{HuPER’s emissions consistently favor the realized sequence on the same acoustics, indicating less emission-level canonical ``auto-correction'' than an XLSR baseline.}
We evaluate controlled contrasts in three settings where canonical restoration is tempting: \textbf{glottalization}, \textbf{flaps}, and \textbf{stop reductions in clusters}.
For each waveform, we compare CTC evidence for a canonical G2P sequence versus a manually verified realized sequence using a length-normalized preference score.
Figure~\ref{fig:phfeat_rsa_centroid} shows that the XLSR baseline is more often canonical-restoring (especially for glottalization and cluster reductions), while \textbf{HuPER remains more acoustic-faithful}.
This matches our design goal: bottom-up emissions should track available cues, and any canonical restoration should be applied explicitly by higher-level constraints when needed.
Diagnostic-set construction, the exact score definition, and the full text/label pairs are provided in Appendix~\ref{app:emission_diag}.


\section{Conclusion and Limitations}
In this work, we view phonetic perception as an adaptive and controllable inference process and introduce HuPER as an initial step toward explicit, human-centered modeling. Our results highlight the value of combining phonetic representations, structured constraints, and dynamic control for robust and interpretable perception. However, the current system is still limited by the small amount of human-verified spoken-phone labels, relies on heuristic and rule-based routing, and models only a small subset of the inference pathways involved in real human perception. Addressing these limitations with larger annotated datasets, learned control policies, and more expressive perceptual models is an important direction for future work.



\bibliography{mybib}
\bibliographystyle{nips}

\newpage
\appendix

\section{Related Work}

\paragraph{Phonetic recognition models}


Traditional phonetic models have been closely co-developed with ASR paradigms and are primarily data-driven, relying on lexicon-derived targets and large-scale supervision, evolving from HMM–GMM~\cite{rabiner2002tutorial,mark2024application} systems to hybrid HMM–DNN~\cite{hinton2012deep} and end-to-end architectures~\cite{graves2006connectionist}. Such systems~\cite{li2020universal,glocker23_interspeech} assign canonical phonemes rather than realized phones, and progress in phone recognition has therefore depended heavily on costly manual annotations. In parallel, linguistically driven approaches incorporate phonological and articulatory priors~\cite{mortensen2016panphon}, such as distinctive features and allophonic modeling, to improve generalization and cross-lingual transfer, but rely on curated linguistic resources and handcrafted representations, limiting their scalability. From a cognitive perspective, phoneme perception itself is a multi-path, closed-loop process that integrates bottom-up acoustic evidence with top-down linguistic and contextual expectations~\cite{mcclelland1986trace, norris2003perceptual}. To date, these three paradigms remain largely disconnected.

\paragraph{Explicit modeling of human-inspired speech perception} 

Early expert systems~\cite{liao2005expert-review} attempted to emulate human cognition through rigid pipelines, which were later shown not to be faithful proxies~\cite{international2025brain}, as cognition emerges from dynamic interactions among functional modules. While cognitive language models, such as~\cite{alkhamissi2025mixture}, introduce brain-inspired modular specialization, these designs remain domain-specific and are not tailored to speech perception. In the context of speech modeling, audio agents~\cite{huang2024audiogpt,liu2025wavjourney} focus on high-level task orchestration and act mainly as wrappers around existing models, without explicitly modeling perceptual mechanisms. Speech world models~\cite{zhou2025speech-world-model} represent early attempts to modularize speech perception, but operate primarily at the prompt level rather than directly decomposing perceptual representations.

\newcommand{\E}{\mathbb{E}}
\newcommand{\Pp}{\mathbb{P}}
\newcommand{\1}{\mathbf{1}}
\newcommand{\R}{\mathbb{R}}
\newcommand{\DeltaK}{\Delta_K}

\section{Doubly Robust Consistency of HuPER Corrector}

\subsection{Setup and notation}
Fix $K\geq 2$. Let $W=(X,Z,Y,\hat Y,A)$ be a generic draw, where:
(i) $X$ are features used by a predictive model $p_\theta(\cdot\mid X)$,
(ii) $Z$ are covariates governing label missingness,
(iii) $Y\in\{1,\dots,K\}$ is the true class label,
(iv) $\hat Y\in\{1,\dots,K\}$ is an always-observed proxy label,
(v) $A\in\{0,1\}$ indicates whether $Y$ is observed (so $AY$ is observed, but $Y$ may be missing when $A=0$). 
Let $e(k)\in\R^K$ denote the one-hot vector with a $1$ in coordinate $k$.
Write $e(Y)$ and $e(\hat Y)$ for the corresponding (random) one-hot vectors of $Y$ and $\hat Y$. Define the multiclass log-loss for any label vector $q\in\R^K$ by
\begin{align*}
\ell_\theta(q,x) := -\sum_{k=1}^K q_k \log p_\theta(k\mid x).
\end{align*}
In particular, $\ell_\theta(e(Y),X)=-\log p_\theta(Y\mid X)$.
The target population risk is
\begin{align*}
R(\theta) := \E\big[\ell_\theta(e(Y),X)\big].
\end{align*}

Define the score vector $s_\theta(X)\in\R^K$ by $s_{\theta,k}(X):=-\log p_\theta(k\mid X)$, so that
\begin{align*}
\ell_\theta(q,X) = q^\top s_\theta(X).
\end{align*}

\paragraph{Missingness model and nuisance functions.}
Define the true propensity (missingness) function
\begin{align*}
g^*(z,\hat y) := \Pp(A=1\mid Z=z,\hat Y=\hat y).
\end{align*}
Define also the conditional class distribution given the always-observed variables
$V:=(X,Z,\hat Y)$:
\begin{align*}
f^*(x,z,\hat y) := \E\big[e(Y)\mid X=x,Z=z,\hat Y=\hat y\big]\in\DeltaK,
\end{align*}
where $\DeltaK:=\big\{q\in\R^K: q_k\geq 0,\ \sum_{k=1}^K q_k=1\big\}$.

\paragraph{Proxy-baseline AIPW corrector.}
For any measurable $g: \mathcal{Z}\times\{1,\dots,K\}\to(0,1]$, define
\begin{align*}
C_g(W) := e(\hat Y) + \frac{A}{g(Z,\hat Y)}\big(e(Y)-e(\hat Y)\big).
\end{align*}
Note that $C_g(W)$ need not lie in $\DeltaK$ (it can have negative entries when $A=1$ and $g(Z,\hat Y)<1$),
but $\ell_\theta(\cdot,X)$ is linear in its first argument, so $\ell_\theta(C_g(W),X)$ is well-defined.

\paragraph{Cross-fitted risk estimator.}
Let $W_1,\dots,W_n$ be i.i.d.\ copies of $W$.
Split $\{1,\dots,n\}$ into $J\geq 2$ folds $I_1,\dots,I_J$ with $|I_j|\to\infty$ and $J$ fixed.
For each fold $j$, fit an estimator $\hat g^{(-j)}$ using only data in $\{W_i:i\notin I_j\}$,
and enforce $\hat g^{(-j)}(z,\hat y)\in[\varepsilon,1]$ for some fixed $\varepsilon\in(0,1)$ (e.g.\ by clipping).
For each $i\in I_j$, set $\hat g_i := \hat g^{(-j)}$ and $C_i := C_{\hat g_i}(W_i)$.
Define the empirical risk
\begin{align*}
\hat R_n(\theta) := \frac{1}{n}\sum_{i=1}^n \ell_\theta(C_i,X_i).
\end{align*}

\subsection{Main result}

We will work on the following assumptions standard in semi-parametric statistics~\citep{kang2007demystifying,wager2018estimation,angelopoulos2023prediction}.

\begin{assumption}[MAR depending only on $(Z,\hat Y)$]\label{ass:mar}
$\Pp(A=1\mid X,Z,Y,\hat Y)=g^*(Z,\hat Y)$ almost surely.
\end{assumption}

\begin{assumption}[Positivity]\label{ass:pos}
There exists $\varepsilon>0$ such that $g^*(Z,\hat Y)\geq \varepsilon$ almost surely.
\end{assumption}

\begin{assumption}[Bounded log score at fixed $\theta$]\label{ass:bound}
There exists $M_\theta<\infty$ such that $\max_{1\leq k\leq K} |s_{\theta,k}(X)| \leq M_\theta$ almost surely.
Equivalently, $\max_{1\leq k\leq K} |\log p_\theta(k\mid X)|\leq M_\theta$ almost surely.
\end{assumption}

\begin{assumption}[Cross-fitted range constraint]\label{ass:clip}
For each fold $j$, the fitted $\hat g^{(-j)}$ satisfies $\hat g^{(-j)}(z,\hat y)\in[\varepsilon,1]$ for all $(z,\hat y)$.
\end{assumption}

We will show consistency under either of the following alternative conditions:
\begin{align*}
\text{(Prop) Propensity consistency:}\quad & \E\big[|\hat g^{(-j)}(Z,\hat Y)-g^*(Z,\hat Y)|\big]\to 0 \text{ for each } j,\\
\text{(Proxy) Proxy label consistency:}\quad & \Pp(\hat Y\ne Y)\to 0.
\end{align*}

\begin{theorem}[Proxy-baseline AIPW corrector is doubly robust]\label{thm:dr}
Assume \ref{ass:mar}--\ref{ass:clip} and fix $\theta$.

\noindent\textbf{1. Exact population bias identity (double robustness).}
For any measurable $g:\mathcal{Z}\times\{1,\dots,K\}\to[\varepsilon,1]$,
\begin{align*}
\E\big[\ell_\theta(C_g(W),X)\big] - R(\theta)
=
\E\left[
\frac{g(Z,\hat Y)-g^*(Z,\hat Y)}{g(Z,\hat Y)}
\big(e(\hat Y)-f^*(X,Z,\hat Y)\big)^\top s_\theta(X)
\right].
\end{align*}
In particular:
\begin{itemize}
\item If $g=g^*$ almost surely, then $\E[\ell_\theta(C_g(W),X)]=R(\theta)$ exactly.
\item If $\hat Y=Y$ almost surely, then $f^*(X,Z,\hat Y)=e(\hat Y)$ almost surely and hence
$\E[\ell_\theta(C_g(W),X)]=R(\theta)$ for any $g$ bounded away from $0$.
\end{itemize}

\noindent\textbf{2. Doubly robust consistency of the cross-fitted risk estimator.}
The cross-fitted estimator $\hat R_n(\theta)$ satisfies
\begin{align*}
\hat R_n(\theta) - R(\theta) \to 0 \quad\text{in probability as } n\to\infty,
\end{align*}
provided either (Prop) holds or (Proxy) holds.
\end{theorem}

\begin{proof}

By definition of $s_\theta(X)$, we write
\begin{align*}
\ell_\theta(q,X) = q^\top s_\theta(X)\quad\text{for all }q\in\R^K.
\end{align*}
Using the definition of $C_g(W)$ and the above, we have
\begin{align}\label{eq:expand-loss}
\ell_\theta(C_g(W),X)
&= \ell_\theta(e(\hat Y),X)
+ \frac{A}{g(Z,\hat Y)}\Big(\ell_\theta(e(Y),X)-\ell_\theta(e(\hat Y),X)\Big).
\end{align}

Let $V:=(X,Z,\hat Y)$ and define
\begin{align*}
\Delta := \ell_\theta(e(Y),X)-\ell_\theta(e(\hat Y),X).
\end{align*}
Under Assumption \ref{ass:mar},
\begin{align*}
\E[A\mid V,Y] = \E[A\mid X,Z,Y,\hat Y] = g^*(Z,\hat Y),
\end{align*}
which implies $A$ is conditionally independent of $Y$ given $V$.
Therefore, by iterated expectations,
\begin{align*}
\E[A\Delta\mid V]
= \E\big[\E[A\Delta\mid V,Y]\mid V\big]
= \E\big[\E[A\mid V,Y]\Delta\mid V\big]
= g^*(Z,\hat Y)\E[\Delta\mid V].
\end{align*}
Hence,
\begin{align}\label{eq:step-3}
\E\left[\frac{A}{g(Z,\hat Y)}\Delta\ \middle|\ V\right]
=
\frac{g^*(Z,\hat Y)}{g(Z,\hat Y)}\E[\Delta\mid V].
\end{align}

Since $s_\theta(X)$ is $V$-measurable (it depends only on $X$),
\begin{align*}
\E\big[\ell_\theta(e(Y),X)\mid V\big]
&= \E\big[e(Y)^\top s_\theta(X)\mid V\big]
= \E[e(Y)\mid V]^\top s_\theta(X)
= f^*(V)^\top s_\theta(X)
= \ell_\theta(f^*(V),X).
\end{align*}
Consequently,
\begin{align}\label{eq:step-4}
R(\theta)
= \E\big[\ell_\theta(e(Y),X)\big]
= \E\big[\E[\ell_\theta(e(Y),X)\mid V]\big]
= \E\big[\ell_\theta(f^*(V),X)\big].
\end{align}

Take conditional expectations of the expansion in Eq.~\eqref{eq:expand-loss} given $V$ and use Eq.~\eqref{eq:step-3}:
\begin{align*}
\E[\ell_\theta(C_g(W),X)\mid V]
&= \ell_\theta(e(\hat Y),X)
+ \frac{g^*(Z,\hat Y)}{g(Z,\hat Y)}\Big(\E[\ell_\theta(e(Y),X)\mid V]-\ell_\theta(e(\hat Y),X)\Big).
\end{align*}
Using Eq.~\eqref{eq:step-4} to substitute $\E[\ell_\theta(e(Y),X)\mid V]=\ell_\theta(f^*(V),X)$ gives
\begin{align*}
\E[\ell_\theta(C_g(W),X)\mid V]
&= \ell_\theta(e(\hat Y),X)
+ \frac{g^*(Z,\hat Y)}{g(Z,\hat Y)}\Big(\ell_\theta(f^*(V),X)-\ell_\theta(e(\hat Y),X)\Big)\\
&= \ell_\theta(f^*(V),X)
+ \left(1-\frac{g^*(Z,\hat Y)}{g(Z,\hat Y)}\right)\Big(\ell_\theta(e(\hat Y),X)-\ell_\theta(f^*(V),X)\Big).
\end{align*}
By linearity,
\begin{align*}
\ell_\theta(e(\hat Y),X)-\ell_\theta(f^*(V),X)
= \big(e(\hat Y)-f^*(V)\big)^\top s_\theta(X).
\end{align*}
Therefore,
\begin{align}\label{eq:step-5}
\E[\ell_\theta(C_g(W),X)\mid V] - \ell_\theta(f^*(V),X)
=
\frac{g(Z,\hat Y)-g^*(Z,\hat Y)}{g(Z,\hat Y)}
\big(e(\hat Y)-f^*(V)\big)^\top s_\theta(X).
\end{align}
Taking unconditional expectations and using $R(\theta)=\E[\ell_\theta(f^*(V),X)]$ yields the stated bias identity.

The first claim follows immediately:
if $g=g^*$ then the multiplicative factor is $0$ almost surely;
if $\hat Y=Y$ almost surely then $Y$ is $V$-measurable and hence $f^*(V)=\E[e(Y)\mid V]=e(\hat Y)$ almost surely.

Under Assumption \ref{ass:clip}, $g(Z,\hat Y)\geq \varepsilon$.
Since $\Vert e(Y)\Vert _1=1$ and $\Vert e(Y)-e(\hat Y)\Vert _1\leq 2$,
\begin{align*}
\Vert C_g(W)\Vert _1
\leq \Vert e(\hat Y)\Vert _1 + \frac{A}{g(Z,\hat Y)}\Vert e(Y)-e(\hat Y)\Vert _1
\leq 1 + \frac{2}{\varepsilon}.
\end{align*}
By Assumption \ref{ass:bound}, $\Vert s_\theta(X)\Vert _\infty\leq M_\theta$ almost surely, so
\begin{align}\label{eq:step-6}
|\ell_\theta(C_g(W),X)|
= |C_g(W)^\top s_\theta(X)|
\leq \Vert C_g(W)\Vert _1\Vert s_\theta(X)\Vert _\infty
\leq \left(1+\frac{2}{\varepsilon}\right)M_\theta.
\end{align}

Fix a fold $j$ and write the fold average
\begin{align*}
\hat R_{n,j}(\theta) := \frac{1}{|I_j|}\sum_{i\in I_j}\ell_\theta(C_{\hat g^{(-j)}}(W_i),X_i).
\end{align*}
Condition on the training data used to fit $\hat g^{(-j)}$ (i.e.\ on $\{W_i:i\notin I_j\}$).
By cross-fitting, the summands for $i\in I_j$ are i.i.d.\ given the training data, and bounded by Eq.~\eqref{eq:step-6}.
Therefore, Chebyshev's inequality (conditional on the training data) yields
\begin{align*}
\hat R_{n,j}(\theta) - \E\big[\ell_\theta(C_{\hat g^{(-j)}}(W),X)\mid \{W_i:i\notin I_j\}\big] \to 0
\quad\text{in probability}.
\end{align*}
Averaging over $j=1,\dots,J$ and using that $J$ is fixed gives
\begin{align}\label{eq:step-7}
\hat R_n(\theta) - \bar\mu_n(\theta) \to 0 \quad\text{in probability},
\end{align}
where
\begin{align*}
\bar\mu_n(\theta) := \sum_{j=1}^J \frac{|I_j|}{n}
\E\big[\ell_\theta(C_{\hat g^{(-j)}}(W),X)\mid \{W_i:i\notin I_j\}\big].
\end{align*}

Fix $j$. Conditional on the training data, $\hat g^{(-j)}$ is deterministic and belongs to $[\varepsilon,1]$.
Apply the population identity from Eq.~\eqref{eq:step-5} with $g=\hat g^{(-j)}$ and then take absolute values:
\begin{align*}
&~\Big|
\E\big[\ell_\theta(C_{\hat g^{(-j)}}(W),X)\mid \{W_i:i\notin I_j\}\big] - R(\theta)
\Big|\\
\leq &~
\E\left[
\left|\frac{\hat g^{(-j)}(Z,\hat Y)-g^*(Z,\hat Y)}{\hat g^{(-j)}(Z,\hat Y)}\right|
\big|\big(e(\hat Y)-f^*(V)\big)^\top s_\theta(X)\big|
\ \middle|\ \{W_i:i\notin I_j\}\right]\\
\leq &~
\frac{M_\theta}{\varepsilon}
\E\left[
|\hat g^{(-j)}(Z,\hat Y)-g^*(Z,\hat Y)|
\Vert e(\hat Y)-f^*(V)\Vert _1
\ \middle|\ \{W_i:i\notin I_j\}\right],
\end{align*}
using $\Vert s_\theta(X)\Vert _\infty\leq M_\theta$ and $\hat g^{(-j)}\ge\varepsilon$.
Since $\Vert e(\hat Y)-f^*(V)\Vert _1\leq 2$, we further have
\begin{align*}
\Big|
\E\big[\ell_\theta(C_{\hat g^{(-j)}}(W),X)\mid \{W_i:i\notin I_j\}\big] - R(\theta)
\Big|
\leq
\frac{2M_\theta}{\varepsilon}
\E\left[
|\hat g^{(-j)}(Z,\hat Y)-g^*(Z,\hat Y)|
\ \middle|\ \{W_i:i\notin I_j\}\right].
\end{align*}
Thus, under condition (Prop) we get the conditional mean converges to $R(\theta)$ in probability.

Under condition (Proxy), we instead use the \emph{exact} identity
\begin{align*}
\E\big[\Vert e(\hat Y)-f^*(V)\Vert _1\big] = 2\Pp(\hat Y\ne Y),
\end{align*}
due to the following: letting $j:=\hat Y$ and writing $f^*_k(V)=\Pp(Y=k\mid V)$,
\begin{align*}
\Vert e(\hat Y)-f^*(V)\Vert _1
&= |1-f^*_j(V)|+\sum_{k\ne j}|0-f^*_k(V)|
= (1-f^*_j(V))+\sum_{k\ne j} f^*_k(V)
= 2(1-f^*_j(V))\\
&= 2\big(1-\Pp(Y=\hat Y\mid V)\big).
\end{align*}
Taking expectations yields $2(1-\E[\Pp(Y=\hat Y\mid V)])=2(1-\Pp(Y=\hat Y))=2\Pp(\hat Y\ne Y)$.
Hence (Proxy) implies $\E[\Vert e(\hat Y)-f^*(V)\Vert _1]\to 0$.
Returning to the bound, the same consistency statement follows due to $|\hat g^{(-j)}-g^*|\leq 1$ and $\hat g^{(-j)}\ge\varepsilon$. 
Combining across folds leads to $\bar\mu_n(\theta)-R(\theta)\to 0$ in probability under either (Prop) or (Proxy).

With the above step and Eq.~\eqref{eq:step-7}, we have
\begin{align*}
\hat R_n(\theta)-R(\theta)
=
\big(\hat R_n(\theta)-\bar\mu_n(\theta)\big)
+
\big(\bar\mu_n(\theta)-R(\theta)\big) \to 0 \quad\text{in probability},
\end{align*}
thus establishing the consistency claim.
\end{proof}

\section{Algorithm}

\begin{algorithm}[H]
  \caption{HuPER self-learning strategy}
  \label{alg:huper_self_learning}
  \begin{algorithmic}
    \STATE {\bfseries Input:} human-labeled phone dataset $\mathcal{D}=\{(X_i,Y_i)\}$; transcript-only dataset $\mathcal{D}'=\{(X_j,T_j)\}$; G2P function $\mathrm{G2P}(\cdot)$; Corrector $C(\cdot)$; number of rounds $R$
    \STATE {\bfseries Output:} HuPER-Recognizer $f^{(R)}$
    \STATE Train initial recognizer $f^{(0)}$ on $\mathcal{D}$ with CTC loss.
    \FOR{$r = 0$ {\bfseries to} $R-1$}
      \STATE Initialize pseudo-labeled set $\tilde{\mathcal{D}}'^{(r)} = \emptyset$.
      \FOR{each $(X,T)$ {\bfseries in} $\mathcal{D}'$}
        \STATE $Z \leftarrow \mathrm{G2P}(T)$ \quad (canonical phoneme sequence)
        \STATE $\bar{Y}^{(r)} \leftarrow f^{(r)}(X)$ \quad (teacher phone prediction)
        \STATE $\tilde{Y}^{(r)} \leftarrow C(X, Z;\bar{Y}^{(r)})$ \quad (corrected pseudo phones; $\bar{Y}^{(r)}$ optional)
        \STATE $\tilde{\mathcal{D}}'^{(r)} \leftarrow \tilde{\mathcal{D}}'^{(r)} \cup \{(X,\tilde{Y}^{(r)})\}$
      \ENDFOR
      \STATE Train $f^{(r+1)}$ on $\tilde{\mathcal{D}}'^{(r)}$ (optionally together with $\mathcal{D}$).
    \ENDFOR
    \STATE \textbf{return} $f^{(R)}$
  \end{algorithmic}
\end{algorithm}

\begin{algorithm}[H]
  \caption{HuPER-Scheduler: distortion-controlled reference-constrained phone refinement}
  \label{alg:huper_scheduler}
  \begin{algorithmic}
    \STATE {\bfseries Input:} utterance $X$; phone evidence graph $\Pi_\theta(X)$; lexicon $L$; LM $G_{\mathrm{LM}}$; threshold $\tau$; optional external reference $R$
    \STATE {\bfseries Output:} final phone sequence $\hat Y$
    \STATE Compute $s(X)$ using Equation~\ref{eq:distortion_score}.
    \IF{$s(X)\le \tau$}
        \STATE $\hat Y \leftarrow \mathrm{Output}(\mathrm{ShortestPath}(\Pi_\theta(X)))$.
    \ELSE
        \IF{$R$ is available}
            \STATE $U \leftarrow R$.
        \ELSE
            \STATE $\hat T \leftarrow \mathrm{Output}\!\left(\mathrm{ShortestPath}\!\left(\Pi_\theta(X)\circ L \circ G_{\mathrm{LM}}\right)\right)$.
            \STATE $U \leftarrow \hat T$.
        \ENDIF
        \STATE Compile Dysfluent WFST constraint $\mathcal{H}\leftarrow \mathcal{H}(U)$.
        \STATE $\hat Y \leftarrow \mathrm{Output}\!\left(\mathrm{ShortestPath}\!\left(\Pi_\theta(X)\circ \mathcal{H}\right)\right)$.
    \ENDIF
    \STATE \textbf{return} $\hat Y$
  \end{algorithmic}
\end{algorithm}

\section{Zero-shot alignment to acoustic and articulatory references}
\label{app:mocha_zero_shot}

\begin{figure}[t]
  \centering

  \begin{subfigure}[t]{0.8\linewidth}
    \centering
    \includegraphics[width=\linewidth]{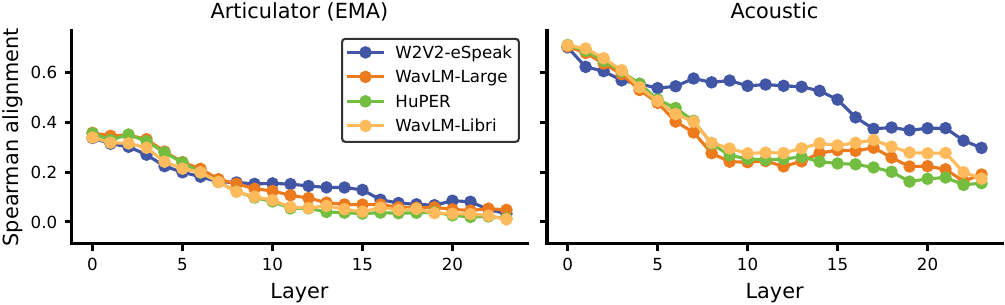}
    \caption{\textbf{Within-phone (same phone).}
    Segment pairs are sampled from the same phone label.
    We report layer-wise Spearman correlation between embedding cosine distances and reference cosine distances
    (EMA vs.\ acoustic), after subtracting a shuffled-reference baseline.}
    \label{fig:mocha_within}
  \end{subfigure}

  \vspace{2mm}

  \begin{subfigure}[t]{0.8\linewidth}
    \centering
    \includegraphics[width=\linewidth]{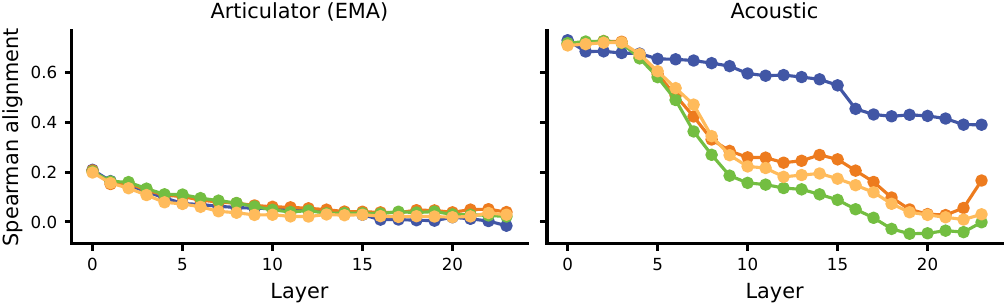}
    \caption{\textbf{Overall (mostly different phones).}
    Segment pairs are sampled uniformly from all segments (thus mainly different phone labels).
    Same metric as above.}
    \label{fig:mocha_overall}
  \end{subfigure}

  \caption{\textbf{Zero-shot alignment to acoustic and articulatory references on MOCHA-TIMIT.}
  Articulatory (EMA) alignment is weak and not diagnostic, and raw acoustic alignment provides limited separation among models.}
  \label{fig:mocha_zero_shot}
\end{figure}

To complement Sec.~\ref{subsec:analysis_phgeom}, we run a segment-level, zero-shot reference-alignment analysis on MOCHA-TIMIT\footnote{\url{https://www.cstr.ed.ac.uk/research/projects/artic/mocha.html}}, which provides paired audio and EMA measurements.
This analysis asks whether model embeddings preserve similarity structure induced by (i) raw acoustics and (ii) articulatory trajectories.

\paragraph{Reference spaces.}
For each phone segment, we compute an \textbf{acoustic} reference vector by averaging an 80-dim log-mel spectrogram over frames.
We compute an \textbf{articulatory (EMA)} reference vector by averaging the 20-dim EMA coordinates within the segment (after utterance-level de-meaning and keeping frames marked as present).
We $\ell_2$-normalize reference vectors and use cosine distance.

\paragraph{Embedding distances and RSA.}
For each model and transformer layer, we mean-pool hidden states over frames inside each phone segment to obtain a segment embedding, then $\ell_2$-normalize it and compute cosine distances on sampled segment pairs.
We report Spearman correlation between embedding distances and reference distances, and subtract a shuffled-reference correlation as a small random baseline (so values near zero indicate no reliable alignment).

\paragraph{Within-phone vs.\ overall pairs.}
We evaluate two pairing regimes:
\textbf{within-phone} samples pairs of segments that share the \emph{same} phone label (probing variability across different realizations of the same phone),
while \textbf{overall} samples pairs uniformly from all segments (mostly \emph{different} phones, probing global geometry).

\paragraph{Key observations.}
Figure~\ref{fig:mocha_within}--\ref{fig:mocha_overall} supports two takeaways.

First, \textbf{articulatory (EMA) alignment is weak}.
Across both within-phone and overall regimes, EMA alignment is modest in the earliest layers and quickly collapses toward zero in mid/late layers, with little separation between models.
This indicates that, under a simple segment-mean EMA summary, embeddings do not preserve a stable articulatory similarity geometry across depth.

Second, \textbf{raw acoustic alignment is also not diagnostic for comparing models}.
All models show similarly high early-layer alignment to log-mel similarity, and the curves provide limited separation among HuPER, WavLM-Raw, and WavLM-Libri (even when trends diverge later).
As a result, a pure acoustic reference does not clearly reveal where HuPER's improvements come from.
This motivates our main-paper focus on a more controlled \emph{acoustic--phonetic} reference (PanPhon distinctive-feature geometry), which is less entangled with nuisance factors and yields clearer, more interpretable cross-model differences.









\section{Phonetic Feature Error Rate (PFER)}
\label{app:pfer}

Following prior work, we use PFER~\cite{mortensen2016panphon}, an articulatory-feature edit distance based on PanPhon distinctive features.
Let $\mathrm{feat}(p)\in\{0,1\}^{24}$ denote the 24-dimensional binary distinctive-feature vector for phone $p$.
For a substitution $p\!\rightarrow\! q$, the cost is the normalized Hamming distance in feature space:
\begin{equation}
d_{\mathrm{sub}}(p,q)=\frac{1}{24}\left\lVert \mathrm{feat}(p)-\mathrm{feat}(q)\right\rVert_1.
\end{equation}
Insertions and deletions each have unit cost.
Let $D(\hat{Y},Y)$ be the minimum-cost Levenshtein distance between the predicted phone sequence $\hat{Y}$ and the reference sequence $Y$ under these costs.
We normalize by the reference length:
\begin{equation}
\mathrm{PFER}(\hat{Y},Y)=\frac{1}{|Y|}\,D(\hat{Y},Y).
\end{equation}
This definition yields a PFER fraction. We multiply it by 100 when reporting percentages in Table~\ref{tab:pfer_results}; the remaining PFER tables and figures report the fraction directly.
Compared to exact-match PER, PFER assigns partial credit to substitutions that differ in only a small number of distinctive features (e.g., voicing).

\newcommand{\err}[1]{\textcolor{red}{#1}}

\section{Failure cases study}
\label{sec:failure_cases}

We analyze failure modes on PPA to understand when (i) bottom-up phone evidence collapses, (ii) WFST-based top-down refinement hurts, and (iii) distortion-based switching makes suboptimal routing decisions.
For each utterance we consider three routes available at test time: \emph{1-best} (bottom-up), \emph{refine} (WFST refinement conditioned on the predicted transcript), and \emph{switch} (distortion-controlled selection between 1-best and refine).
For diagnosis only, we additionally report \emph{refine$_{given}$}, i.e., refinement conditioned on an external reference transcript (not available at test time).

\subsection{Case selection protocol}
\label{sec:failure_select}
We select representative cases using three reproducible criteria: (1) highest PFER under 1-best, (2) largest degradation from refinement ($\Delta=\mathrm{PFER}_{ref}-\mathrm{PFER}_{1b}$), and (3) largest \emph{switching regret} defined as $\mathrm{PFER}_{sw}-\min(\mathrm{PFER}_{1b},\mathrm{PFER}_{ref})$.
To compute the switching route in this appendix, we use a fixed threshold $\tau=0.573$ on the distortion score (chosen to minimize mean PFER on this evaluation set).
We omit utterance identifiers and audio references for privacy.

\subsection{Failure taxonomy}
\label{sec:failure_taxonomy}
We categorize failures into four practical types:
\textbf{(A) Extreme evidence failures}, where 1-best exhibits large insertion bursts or very short phone references make the metric unstable;
\textbf{(B) Wrong hypothesis conditioning}, where refine hurts but refine$_{given}$ would help, indicating that the guiding hypothesis (predicted transcript) is incorrect;
\textbf{(C) Over-constrained top-down}, where both refine and refine$_{given}$ hurt, suggesting pronunciation coverage or LM bias issues;
and \textbf{(D) Distortion/scheduler outliers}, where distortion-based routing makes noticeable mistakes (false positives/negatives) or distortion is weakly aligned with PFER for that sample.

\subsection{Representative examples and aggregate summary}
\label{sec:failure_examples}
Table~\ref{tab:ppa_failure_cases} lists representative examples for each failure type, and Table~\ref{tab:ppa_failure_summary} summarizes aggregate statistics by category.

\subsection{Implications}
\label{sec:failure_implications}
Across (B) and (C), a recurring pattern is that hard constraints can hurt unless the guiding hypothesis is reliable and the pronunciation model has sufficient coverage for weak/atypical realizations.
In practice, we found the following mitigations useful: (i) hypothesis ensembling (e.g., ASR $N$-best) when conditioning refinement; (ii) softening constraints by tuning LM weight/word insertion penalty and expanding pronunciation variants; and (iii) sanity checks to flag extremely short references or insertion bursts before interpreting PFER.

\begin{table}[t]
\centering
\small
\setlength{\tabcolsep}{5pt}
\begin{tabular}{l r c c c c c}
\toprule
Category & $n$ & Dist. & PFER$_{1b}$ & PFER$_{ref}$ & $\Delta$ & Regret \\
\midrule
Extreme evidence failures & 5 & 0.614 & 2.214 & 1.023 & -1.191 & 0.000 \\
Wrong hypothesis conditioning & 11 & 0.570 & 0.277 & 0.551 & 0.274 & 0.093 \\
Over-constrained top-down & 16 & 0.563 & 0.223 & 0.422 & 0.199 & 0.089 \\
Distortion/scheduler outliers & 5 & 0.512 & 0.751 & 0.585 & -0.167 & 0.174 \\
Other & 37 & 0.543 & 0.290 & 0.305 & 0.016 & 0.018 \\
\bottomrule
\addlinespace[4pt]
\end{tabular}
\caption{Aggregate statistics by failure category. $\Delta=\mathrm{PFER}_{ref}-\mathrm{PFER}_{1b}$. Regret is computed for distortion switching with $\tau=0.573$.}
\label{tab:ppa_failure_summary}
\end{table}

\begin{table}[t]
\centering
\small
\setlength{\tabcolsep}{3pt}
\resizebox{\linewidth}{!}{
\begin{tabular}{l p{5.4cm} c c c c c p{4.8cm}}
\toprule
Case & Ref text (abbr.) & Dist. & PFER$_{1b}$ & PFER$_{ref}$ & PFER$_{given}$ & PFER$_{sw}$ & Diagnosis \\
\midrule
\multicolumn{8}{l}{\textbf{A. Extreme evidence failures}}\\
A1 & Twice each day he plays skillfully and with zest \ldots & 0.597 & 3.638 & 2.350 & 2.258 & 2.350 & 1-best insertion burst; refine stabilizes. \\
A2 & yet he still thinks as swiftly as ever & 0.643 & 3.225 & 0.250 & 1.800 & 0.250 & Short phone ref; insertion dominates. \\
A3 & Grandfather likes to be modern in his language & 0.656 & 1.783 & 1.022 & 0.395 & 1.022 & 1-best insertion burst; refine stabilizes. \\
A4 & giving those who observe him a pronounced feeling of \ldots & 0.593 & 1.359 & 0.915 & 0.859 & 0.915 & Both routes fail; weak evidence. \\
\midrule
\multicolumn{8}{l}{\textbf{B. Wrong hypothesis conditioning}}\\
B1 & Well he is nearly ninety three years old & 0.649 & 0.357 & 0.663 & 0.312 & 0.663 & Hypothesis mismatch: ref-conditioned refine hurts; Given ref helps. \\
B2 & When he speaks & 0.496 & 0.163 & 0.483 & 0.208 & 0.163 & Hypothesis mismatch: ref-conditioned refine hurts; Given ref helps. \\
B3 & He dresses himself in an old black frock coat & 0.568 & 0.243 & 0.635 & 0.364 & 0.243 & Hypothesis mismatch: ref-conditioned refine hurts; Given ref helps. \\
B4 & He dresses himself in an old black frock coat & 0.605 & 0.430 & 0.765 & 0.500 & 0.765 & Hypothesis mismatch: ref-conditioned refine hurts; Given ref helps. \\
\midrule
\multicolumn{8}{l}{\textbf{C. Over-constrained top-down}}\\
C1 & but he always answers Banana oil & 0.585 & 0.150 & 0.556 & 0.679 & 0.556 & Over-constraint: refine hurts even with Given ref. \\
C2 & but he always answers Banana oil & 0.617 & 0.361 & 0.703 & 0.696 & 0.703 & Over-constraint: refine hurts even with Given ref. \\
C3 & A long beard clings to his chin & 0.560 & 0.242 & 0.554 & 0.554 & 0.242 & Over-constraint: refine hurts even with Given ref. \\
C4 & Yet he still thinks as swiftly as ever & 0.643 & 0.286 & 0.591 & 0.535 & 0.591 & Over-constraint: refine hurts even with Given ref. \\
\midrule
\multicolumn{8}{l}{\textbf{D. Distortion/scheduler outliers}}\\
D1 & usually several buttons are missing & 0.589 & 0.240 & 0.619 & 0.619 & 0.619 & Scheduler FP: high distortion but refine hurts. \\
D2 & but he always answers Banana oil & 0.566 & 0.657 & 0.361 & 0.404 & 0.657 & Scheduler FN: low distortion but refine would help. \\
D3 & giving those who observe him a pronounced feeling of \ldots & 0.553 & 0.951 & 0.656 & 0.719 & 0.951 & Scheduler FN: low distortion but refine would help. \\
D4 & he slowly takes a short walk in the open \ldots & 0.519 & 0.663 & 0.422 & 0.422 & 0.663 & Scheduler FN: low distortion but refine would help. \\
\bottomrule
\addlinespace[4pt]
\end{tabular}
}
\caption{Representative failure cases on PPA. Dist. is the distortion score. PFER$_{1b}$ is HuPER 1-best, PFER$_{ref}$ is refinement conditioned on the predicted transcript, PFER$_{given}$ is refinement conditioned on an external reference transcript, and PFER$_{sw}$ is distortion-controlled switching with $\tau=0.573$.}
\label{tab:ppa_failure_cases}
\end{table}

\subsection{Concrete error snippets}
\label{sec:failure_snippets}
Red marks phones/words that are not aligned to the reference (substitutions/insertions). Deleted reference phones are shown with \sout{} on the GT line. A red underscore indicates an empty counterpart in the alignment.

\paragraph{Case A1.}
\noindent\textbf{Reference text:} Twice each day he plays skillfully and with zest upon a small organ\\
\textbf{Predicted hypothesis (for refinement):} twice \textcolor{red}{a} day he \textcolor{red}{place} \textcolor{red}{skilly} and with \textcolor{red}{zeppa} \textcolor{red}{ponce} \textcolor{red}{mogin}\\
\textbf{Distortion:} 0.597.\quad \textbf{PFER}$_{1b}$=3.638, \textbf{PFER}$_{ref}$=2.350, \textbf{PFER}$_{given}$=2.258.
\\\textit{Extreme evidence failure with a large insertion burst; both routes struggle.}
\vspace{2pt}
\begin{tabular}{@{}p{1.6cm}p{0.91\linewidth}@{}}
GT: & {\ttfamily \textcolor{red}{\_} \textcolor{red}{\_} EH \textcolor{red}{\_} \textcolor{red}{\_} \textcolor{red}{\_} \textcolor{red}{\_} N \textcolor{red}{\_} \textcolor{red}{\_} \textcolor{red}{\_} \textcolor{red}{\_} \textcolor{red}{\_} HH IY S \textcolor{red}{\_} P}\\
1-best: & {\ttfamily \textcolor{red}{+K} \textcolor{red}{+IH} EH \textcolor{red}{+L} \textcolor{red}{+L} \textcolor{red}{+IY} \textcolor{red}{+IH} N \textcolor{red}{+D} \textcolor{red}{+W} \textcolor{red}{+IH} \textcolor{red}{+TH} \textcolor{red}{+Z} \textcolor{red}{EH} \textcolor{red}{P} S \textcolor{red}{+AH} P}\\
refine: & {\ttfamily \textcolor{red}{\_} \textcolor{red}{\_} \textcolor{red}{\_} \textcolor{red}{\_} \textcolor{red}{\_} \textcolor{red}{\_} \textcolor{red}{\_} \textcolor{red}{\_} \textcolor{red}{\_} \textcolor{red}{\_} \textcolor{red}{\_} \textcolor{red}{\_} N D \textcolor{red}{\_} \textcolor{red}{\_} \textcolor{red}{\_} \textcolor{red}{\_}}\\
\end{tabular}
\vspace{6pt}

\paragraph{Case A2.}
\noindent\textbf{Reference text:} yet he still thinks as swiftly as ever\\
\textbf{Predicted hypothesis (for refinement):} \textcolor{red}{citizen}\\
\textbf{Distortion:} 0.643.\quad \textbf{PFER}$_{1b}$=3.225, \textbf{PFER}$_{ref}$=0.250, \textbf{PFER}$_{given}$=1.800.
\\\textit{Very short phone reference; refinement can sharply reduce insertion-driven PFER but may be unstable.}
\vspace{2pt}
\begin{tabular}{@{}p{1.6cm}p{0.91\linewidth}@{}}
GT: & {\ttfamily Y EH T HH IY S T IH L TH IH NG K S AE Z S W IH F T L IY AE Z EH V ER}\\
1-best: & {\ttfamily \textcolor{red}{+S} \textcolor{red}{+IH} \textcolor{red}{+T} \textcolor{red}{+IH} \textcolor{red}{+Z} \textcolor{red}{+AH} \textcolor{red}{+N} \textcolor{red}{+CH} \textcolor{red}{+AH} \textcolor{red}{+N} \textcolor{red}{+T} \textcolor{red}{+S} \textcolor{red}{+AH} \textcolor{red}{+N} \textcolor{red}{+CH} \textcolor{red}{+AH} \textcolor{red}{+N} \textcolor{red}{+T}}\\
refine: & {\ttfamily Y EH T HH IY S T IH L TH IH NG K S AE Z S W IH F T L IY AE Z EH V ER}\\
\end{tabular}
\vspace{6pt}

\paragraph{Case B2.}
\noindent\textbf{Reference text:} When he speaks\\
\textbf{Predicted hypothesis (for refinement):} he \textcolor{red}{spake}\\
\textbf{Distortion:} 0.496.\quad \textbf{PFER}$_{1b}$=0.163, \textbf{PFER}$_{ref}$=0.483, \textbf{PFER}$_{given}$=0.208.
\\\textit{Wrong-hypothesis conditioning: refinement is guided by an incorrect word form (``spake''), distorting constraints.}
\vspace{2pt}
\begin{tabular}{@{}p{1.6cm}p{0.91\linewidth}@{}}
GT: & {\ttfamily HH EH N HH IY S P IY K S}\\
1-best: & {\ttfamily HH IY \textcolor{red}{Z} \textcolor{red}{S} P IY K S}\\
refine: & {\ttfamily HH IY \textcolor{red}{S} P \textcolor{red}{EY} K}\\
\end{tabular}
\vspace{6pt}

\paragraph{Case B1.}
\noindent\textbf{Reference text:} Well he is nearly ninety three years old\\
\textbf{Predicted hypothesis (for refinement):} well \textcolor{red}{nigh} \textcolor{red}{thistle}\\
\textbf{Distortion:} 0.649.\quad \textbf{PFER}$_{1b}$=0.357, \textbf{PFER}$_{ref}$=0.663, \textbf{PFER}$_{given}$=0.312.
\\\textit{Wrong-hypothesis conditioning: hypothesis is far from the reference, so refinement degrades phone accuracy.}
\vspace{2pt}
\begin{tabular}{@{}p{1.6cm}p{0.91\linewidth}@{}}
GT: & {\ttfamily W EH L HH IY IH Z N IH R L IY N AY N T IY TH R IY Y IH R Z OW L D}\\
1-best: & {\ttfamily W EH L HH IY IH Z N IH R L IY \textcolor{red}{N} AY N \textcolor{red}{S} IY \textcolor{red}{TH} R IY \textcolor{red}{S} \textcolor{red}{AH} L D}\\
refine: & {\ttfamily W EH L \textcolor{red}{N} AY \textcolor{red}{TH} \textcolor{red}{AH} \textcolor{red}{S} \textcolor{red}{AH} L}\\
\end{tabular}
\vspace{6pt}

\paragraph{Case C1.}
\noindent\textbf{Reference text:} but he always answers Banana oil\\
\textbf{Predicted hypothesis (for refinement):} \textcolor{red}{business} \textcolor{red}{boil}\\
\textbf{Distortion:} 0.585.\quad \textbf{PFER}$_{1b}$=0.150, \textbf{PFER}$_{ref}$=0.556, \textbf{PFER}$_{given}$=0.679.
\\\textit{Over-constrained top-down: even with the correct transcript, refinement hurts (pronunciation/LM bias).}
\vspace{2pt}
\begin{tabular}{@{}p{1.6cm}p{0.91\linewidth}@{}}
GT: & {\ttfamily B AH T HH IY AO L W EY Z AE N S ER Z B AH N AE N AH OY L}\\
1-best: & {\ttfamily B AH T HH IY AO L W EY Z AE N S ER Z B AH N AE N AH OY L}\\
refine: & {\ttfamily \textcolor{red}{B} \textcolor{red}{IH} \textcolor{red}{Z} \textcolor{red}{N} \textcolor{red}{AH} \textcolor{red}{S} \textcolor{red}{B} \textcolor{red}{OY} L}\\
\end{tabular}
\vspace{6pt}

\paragraph{Case C3.}
\noindent\textbf{Reference text:} A long beard clings to his chin\\
\textbf{Predicted hypothesis (for refinement):} \textcolor{red}{henri} \textcolor{red}{kisses} chin\\
\textbf{Distortion:} 0.560.\quad \textbf{PFER}$_{1b}$=0.242, \textbf{PFER}$_{ref}$=0.554, \textbf{PFER}$_{given}$=0.554.
\\\textit{Over-constrained top-down on a short phrase; refinement introduces consistent substitutions.}
\vspace{2pt}
\begin{tabular}{@{}p{1.6cm}p{0.91\linewidth}@{}}
GT: & {\ttfamily AH L AO NG B IH R D K L IH NG Z T UW HH IH Z CH IH N}\\
1-best: & {\ttfamily AH \textcolor{red}{L} AO NG B IH R D K L IH NG Z T UW HH IH Z CH IH N}\\
refine: & {\ttfamily \textcolor{red}{HH} \textcolor{red}{EH} \textcolor{red}{N} \textcolor{red}{R} IY \textcolor{red}{K} \textcolor{red}{IH} \textcolor{red}{S} \textcolor{red}{AH} Z CH IH N}\\
\end{tabular}
\vspace{6pt}

\paragraph{Case D2.}
\noindent\textbf{Reference text:} but he always answers Banana oil\\
\textbf{Predicted hypothesis (for refinement):} but he \textcolor{red}{has} he \textcolor{red}{was} answers banana\\
\textbf{Distortion:} 0.566.\quad \textbf{PFER}$_{1b}$=0.657, \textbf{PFER}$_{ref}$=0.361.
\\\textit{Scheduler false negative at $\tau=0.573$: distortion is below threshold so switching keeps 1-best, but refinement would help.}
\vspace{2pt}
\begin{tabular}{@{}p{1.6cm}p{0.91\linewidth}@{}}
GT: & {\ttfamily B AH T HH IY AO L W EY Z AE N S ER Z B AH N AE N AH OY L}\\
1-best: & {\ttfamily \textcolor{red}{S} \textcolor{red}{P} B AH \textcolor{red}{DX} HH IY HH AE \textcolor{red}{Z} HH IY \textcolor{red}{W} AA Z AE N S ER Z B AH N AE N AH}\\
refine: & {\ttfamily B AH T HH IY \textcolor{red}{HH} AE Z HH IY W AA Z AE N S ER Z B AH N AE N AH OY L}\\
\end{tabular}
\vspace{6pt}

\paragraph{Case D3.}
\noindent\textbf{Reference text:} giving those who observe him a pronounced feeling of the utmost respect\\
\textbf{Predicted hypothesis (for refinement):} giving those who observe him \textcolor{red}{announced} \textcolor{red}{filling} \textcolor{red}{of} the \textcolor{red}{upmost} respect\\
\textbf{Distortion:} 0.553.\quad \textbf{PFER}$_{1b}$=0.951, \textbf{PFER}$_{ref}$=0.656.
\\\textit{Scheduler false negative at $\tau=0.573$: switching keeps 1-best but refinement reduces PFER under weak evidence.}
\vspace{2pt}
\begin{tabular}{@{}p{1.6cm}p{0.91\linewidth}@{}}
GT: & {\ttfamily G IH V IH NG DH OW Z HH UW AH B Z ER V HH IH M EY P R AH N AW N S T F IY L IH NG AH V DH AH AH T ER M OW S T R IH S P EH K T}\\
1-best: & {\ttfamily G IH V IH NG \textcolor{red}{N} DH OW Z HH UW AH B \textcolor{red}{Z} ER V HH IH M \textcolor{red}{AH} \textcolor{red}{N} \textcolor{red}{UH} N S T F IY \textcolor{red}{L} IH NG AH V \textcolor{red}{DH} AH}\\
refine: & {\ttfamily G IH V IH NG DH OW Z HH UW AH B Z ER V HH IH M EY P R AH N AW N S T F IY L IH NG AH V DH AH}\\
\end{tabular}
\vspace{6pt}



\section{Analysis details for understanding HuPER-Recognizer gains}
\label{app:analysis_details}

\subsection{Centroid RSA: setup and implementation}
\label{app:analysis_rsa}

\paragraph{Goal.}
We quantify whether an encoder’s phone representations are organized by broad acoustic--phonetic similarity, using PanPhon distinctive-feature distances as a controlled proxy~\cite{mortensen2016panphon}.

\paragraph{Data.}
We run the analysis on TIMIT, which provides time-stamped phone segments. Each segment is associated with a phone label.

\paragraph{Phone representations.}
At each transformer layer $\ell$, we extract hidden states and mean-pool within each phone segment to obtain a segment embedding.
We then average segment embeddings of the same phone to form a phone centroid $\mathbf{c}^{(\ell)}(p)$.

\paragraph{Distances and RSA score.}
We compute pairwise cosine distances between phone centroids to form an embedding-distance matrix $D^{(\ell)}_{\text{emb}}$.
We also compute a PanPhon feature-distance matrix $D_{\text{pan}}$ over the same phone set.
The layer-wise RSA score is the Spearman correlation between vectorized upper triangles of $D^{(\ell)}_{\text{emb}}$ and $D_{\text{pan}}$.

\paragraph{Models compared.}
We include (i) \emph{Wav2Vec2Phoneme} (multilingual phone recognition baseline), (ii) \emph{WavLM-Raw} (pretrained WavLM-Large), (iii) \emph{WavLM-Libri} (English-only fine-tuning on LibriSpeech with G2P labels), and (iv) HuPER-Recognizer.


\subsection{Emission diagnostic: construction, scoring, and text inputs}
\label{app:emission_diag}

We use a controlled diagnostic to test whether CTC emissions prefer the \emph{canonical} (G2P) phone sequence or the \emph{realized} phone sequence supported by the acoustics.

\textbf{Diagnostic set.}
We focus on three cases where canonical restoration is tempting in casual English speech:
\textbf{glottalization}, \textbf{flaps}, and \textbf{stops in consonant clusters}.
For each synthesized utterance waveform $x$, we create a matched contrast $(x, y_{\mathrm{can}}, y_{\mathrm{real}})$:
$y_{\mathrm{can}}$ is the canonical G2P phone sequence, and $y_{\mathrm{real}}$ is a realized phone sequence verified by listening.

\textbf{CTC evidence.}
Given a CTC recognizer, we compute the marginal log-likelihood $\log P(y\mid x)$ via the forward algorithm in log-space.
To compare sequences of different lengths, we use a per-phone normalized preference score:
\begin{equation}
\Delta_{\mathrm{norm}}(x)
=
\frac{\log P(y_{\mathrm{can}} \mid x)}{|y_{\mathrm{can}}|}
-
\frac{\log P(y_{\mathrm{real}} \mid x)}{|y_{\mathrm{real}}|}.
\end{equation}
$\Delta_{\mathrm{norm}}>0$ indicates canonical-restoring emissions, while $\Delta_{\mathrm{norm}}<0$ indicates acoustic-faithful emissions favoring the realized sequence.

\textbf{Reporting.}
In the main paper, we plot per-utterance scores for HuPER vs.\ the XLSR baseline (Fig.~\ref{fig:phfeat_rsa_centroid}).
We also report a compact category-wise summary:
\begin{table}[h]
\centering
\setlength{\tabcolsep}{6pt}
\renewcommand{\arraystretch}{1.15}
\resizebox{\linewidth}{!}{
\begin{tabular}{lcc|cc|cc}
\toprule
& \multicolumn{2}{c|}{HuPER-Recognizer (ours)}
& \multicolumn{2}{c|}{XLSR baseline}
& \multicolumn{2}{c}{Paired diff (XLSR -- Ours)} \\
Category
& median $\Delta_{\mathrm{norm}}$ & $\Pr(\Delta_{\mathrm{norm}}>0)$
& median $\Delta_{\mathrm{norm}}$ & $\Pr(\Delta_{\mathrm{norm}}>0)$
& median $\Delta_{\mathrm{norm}}$ & $\Pr(>0)$ \\
\midrule
Glottalization   & $-1.018$ & $0.00$ & $+0.889$ & $0.70$ & $+1.535$ & $1.00$ \\
Flaps            & $-0.768$ & $0.00$ & $-0.459$ & $0.30$ & $+0.630$ & $0.90$ \\
Stops in clusters& $-0.424$ & $0.20$ & $+0.413$ & $0.90$ & $+0.886$ & $0.80$ \\
\bottomrule
\addlinespace[4pt]
\end{tabular}
}
\caption{\textbf{Emission-level diagnostic using $\Delta_{\mathrm{norm}}$.}
$\Delta_{\mathrm{norm}}>0$ indicates that emissions prefer the canonical phone sequence over the realized sequence.
The paired-difference columns compare the two models on the same utterances; positive values mean the XLSR baseline is more canonical-restoring than HuPER-Recognizer.}
\label{tab:emission_delta}
\end{table}

\textbf{Text inputs (plain text only).}
This appendix also records the exact text inputs used to synthesize the diagnostic waveforms referenced in Sec.~\ref{subsec:analysis_emission}.
All inputs are plain text (no style tags or special TTS instructions). We generate short, TTS-friendly phrases
using a single prompt template (run once per category):
\begin{quote}\small\ttfamily
Generate 10 short English phrases (2--4 words) that are likely to be pronounced in casual speech with the following phenomenon: \{PHENOMENON\}. \\
Constraints: common words, no proper nouns, keep it short and natural for TTS. \\
Return only the 10 phrases, one per line.
\end{quote}
The resulting phrases (10 per category) are listed in Table~\ref{tab:emission_text_3col}. After audio generation,
each item is manually checked and paired with $(y_{\mathrm{can}}, y_{\mathrm{real}})$ for the emission test.

\begin{table*}[h]
\centering
\small
\setlength{\tabcolsep}{10pt}
\renewcommand{\arraystretch}{1.12}
\begin{tabularx}{\textwidth}{X X X}
\toprule
Glottalization & Flaps & Stops in consonant clusters \\
\midrule
a button & a better idea & last Sunday \\
my kitten & a little later & next day \\
that mountain & water bottle & just say \\
in Britain & get it & best friend \\
a little bit & put it away & first time \\
not now & what a day & most people \\
can't go & write it down & west side \\
get back & I need it & old man \\
sit down & go to bed & hand bag \\
at night & it is ready & asked to \\
\bottomrule
\end{tabularx}
\caption{Text inputs used to synthesize the emission diagnostic set in Sec.~\ref{subsec:analysis_emission} (plain text only; 10 per category).}
\label{tab:emission_text_3col}
\end{table*}

\newpage

\begin{table*}[h]
\centering
\caption{Evaluation datasets for phone recognition.}
\label{tab:eval_datasets_phone}
\small
\setlength{\tabcolsep}{6pt}
\begin{tabularx}{\textwidth}{l X}
\toprule
Dataset & Description \\
\midrule
Buckeye~\cite{pitt2005buckeye} & Natural English conversational speech with human phonetic annotation. \\
DRC-SE (DoReCo South-England)~\cite{paschen2020doreco} & English dialectal-variation subset from DoReCo (South England). \\
L2-ARCTIC-Perceived~\cite{zhao2018l2arctic} & L2 English speech corpus with human-verified / annotated pronunciations. \\
EpaDB~\cite{vidal2019epadb} & L2 English (Spanish-accented) speech with detailed phonetic annotations (e.g., mispronunciations). \\
SpeechOcean762~\cite{zhang2021speechocean762} & Large-scale L2 English corpus with human annotations / verification. \\
VoxAngeles~\cite{conf/coling/ChodroffPBM24} & Multilingual word recordings; evaluate on languages such as Chamorro, Degema, Lakota, Pampanga, Iloko, etc. \\
\bottomrule
\end{tabularx}
\end{table*}

\begin{longtblr}[
  caption = {\textbf{Zero-shot multilingual phone recognition on VoxAngeles (94 languages).} Per-language PFER is reported as a fraction (lower is better). Entries of 1.00 mean the model lacks support for the corresponding language/inventory. Red highlights languages where HuPER attains the best (lowest) PFER across all baselines.},
  label = {tab:voxangles}
]{
  cell{2}{2} = {r},
  cell{2}{3} = {r},
  cell{2}{4} = {r},
  cell{2}{5} = {r},
  cell{2}{6} = {r},
  cell{2}{7} = {r},
  cell{3}{2} = {r},
  cell{3}{3} = {r},
  cell{3}{4} = {r},
  cell{3}{5} = {r},
  cell{3}{6} = {r},
  cell{3}{7} = {r},
  cell{4}{2} = {r},
  cell{4}{3} = {r},
  cell{4}{4} = {r},
  cell{4}{5} = {r},
  cell{4}{6} = {r},
  cell{4}{7} = {r,fg=red},
  cell{5}{2} = {r},
  cell{5}{3} = {r},
  cell{5}{4} = {r},
  cell{5}{5} = {r},
  cell{5}{6} = {r},
  cell{5}{7} = {r},
  cell{6}{2} = {r},
  cell{6}{3} = {r},
  cell{6}{4} = {r},
  cell{6}{5} = {r},
  cell{6}{6} = {r},
  cell{6}{7} = {r},
  cell{7}{2} = {r},
  cell{7}{3} = {r},
  cell{7}{4} = {r},
  cell{7}{5} = {r},
  cell{7}{6} = {r},
  cell{7}{7} = {r,fg=red},
  cell{8}{2} = {r},
  cell{8}{3} = {r},
  cell{8}{4} = {r},
  cell{8}{5} = {r},
  cell{8}{6} = {r},
  cell{8}{7} = {r},
  cell{9}{2} = {r},
  cell{9}{3} = {r},
  cell{9}{4} = {r},
  cell{9}{5} = {r},
  cell{9}{6} = {r},
  cell{9}{7} = {r,fg=red},
  cell{10}{2} = {r},
  cell{10}{3} = {r},
  cell{10}{4} = {r},
  cell{10}{5} = {r},
  cell{10}{6} = {r},
  cell{10}{7} = {r},
  cell{11}{2} = {r},
  cell{11}{3} = {r},
  cell{11}{4} = {r},
  cell{11}{5} = {r},
  cell{11}{6} = {r},
  cell{11}{7} = {r},
  cell{12}{2} = {r},
  cell{12}{3} = {r},
  cell{12}{4} = {r},
  cell{12}{5} = {r},
  cell{12}{6} = {r},
  cell{12}{7} = {r},
  cell{13}{2} = {r},
  cell{13}{3} = {r},
  cell{13}{4} = {r},
  cell{13}{5} = {r},
  cell{13}{6} = {r},
  cell{13}{7} = {r,fg=red},
  cell{14}{2} = {r},
  cell{14}{3} = {r},
  cell{14}{4} = {r},
  cell{14}{5} = {r},
  cell{14}{6} = {r},
  cell{14}{7} = {r},
  cell{15}{2} = {r},
  cell{15}{3} = {r},
  cell{15}{4} = {r},
  cell{15}{5} = {r},
  cell{15}{6} = {r},
  cell{15}{7} = {r},
  cell{16}{2} = {r},
  cell{16}{3} = {r},
  cell{16}{4} = {r},
  cell{16}{5} = {r},
  cell{16}{6} = {r},
  cell{16}{7} = {r},
  cell{17}{2} = {r},
  cell{17}{3} = {r},
  cell{17}{4} = {r},
  cell{17}{5} = {r},
  cell{17}{6} = {r},
  cell{17}{7} = {r},
  cell{18}{2} = {r},
  cell{18}{3} = {r},
  cell{18}{4} = {r},
  cell{18}{5} = {r},
  cell{18}{6} = {r},
  cell{18}{7} = {r},
  cell{19}{2} = {r},
  cell{19}{3} = {r},
  cell{19}{4} = {r},
  cell{19}{5} = {r},
  cell{19}{6} = {r},
  cell{19}{7} = {r},
  cell{20}{2} = {r},
  cell{20}{3} = {r},
  cell{20}{4} = {r},
  cell{20}{5} = {r},
  cell{20}{6} = {r},
  cell{20}{7} = {r},
  cell{21}{2} = {r},
  cell{21}{3} = {r},
  cell{21}{4} = {r},
  cell{21}{5} = {r},
  cell{21}{6} = {r},
  cell{21}{7} = {r},
  cell{22}{2} = {r},
  cell{22}{3} = {r},
  cell{22}{4} = {r},
  cell{22}{5} = {r},
  cell{22}{6} = {r},
  cell{22}{7} = {r},
  cell{23}{2} = {r},
  cell{23}{3} = {r},
  cell{23}{4} = {r},
  cell{23}{5} = {r},
  cell{23}{6} = {r},
  cell{23}{7} = {r},
  cell{24}{2} = {r},
  cell{24}{3} = {r},
  cell{24}{4} = {r},
  cell{24}{5} = {r},
  cell{24}{6} = {r},
  cell{24}{7} = {r,fg=red},
  cell{25}{2} = {r},
  cell{25}{3} = {r},
  cell{25}{4} = {r},
  cell{25}{5} = {r},
  cell{25}{6} = {r},
  cell{25}{7} = {r},
  cell{26}{2} = {r},
  cell{26}{3} = {r},
  cell{26}{4} = {r},
  cell{26}{5} = {r},
  cell{26}{6} = {r},
  cell{26}{7} = {r},
  cell{27}{2} = {r},
  cell{27}{3} = {r},
  cell{27}{4} = {r},
  cell{27}{5} = {r},
  cell{27}{6} = {r},
  cell{27}{7} = {r},
  cell{28}{2} = {r},
  cell{28}{3} = {r},
  cell{28}{4} = {r},
  cell{28}{5} = {r},
  cell{28}{6} = {r},
  cell{28}{7} = {r},
  cell{29}{2} = {r},
  cell{29}{3} = {r},
  cell{29}{4} = {r},
  cell{29}{5} = {r},
  cell{29}{6} = {r},
  cell{29}{7} = {r},
  cell{30}{2} = {r},
  cell{30}{3} = {r},
  cell{30}{4} = {r},
  cell{30}{5} = {r},
  cell{30}{6} = {r},
  cell{30}{7} = {r,fg=red},
  cell{31}{2} = {r},
  cell{31}{3} = {r},
  cell{31}{4} = {r},
  cell{31}{5} = {r},
  cell{31}{6} = {r},
  cell{31}{7} = {r},
  cell{32}{2} = {r},
  cell{32}{3} = {r},
  cell{32}{4} = {r},
  cell{32}{5} = {r},
  cell{32}{6} = {r},
  cell{32}{7} = {r},
  cell{33}{2} = {r},
  cell{33}{3} = {r},
  cell{33}{4} = {r},
  cell{33}{5} = {r},
  cell{33}{6} = {r},
  cell{33}{7} = {r},
  cell{34}{2} = {r},
  cell{34}{3} = {r},
  cell{34}{4} = {r},
  cell{34}{5} = {r},
  cell{34}{6} = {r},
  cell{34}{7} = {r},
  cell{35}{2} = {r},
  cell{35}{3} = {r},
  cell{35}{4} = {r},
  cell{35}{5} = {r},
  cell{35}{6} = {r},
  cell{35}{7} = {r},
  cell{36}{2} = {r},
  cell{36}{3} = {r},
  cell{36}{4} = {r},
  cell{36}{5} = {r},
  cell{36}{6} = {r},
  cell{36}{7} = {r},
  cell{37}{2} = {r},
  cell{37}{3} = {r},
  cell{37}{4} = {r},
  cell{37}{5} = {r},
  cell{37}{6} = {r},
  cell{37}{7} = {r},
  cell{38}{2} = {r},
  cell{38}{3} = {r},
  cell{38}{4} = {r},
  cell{38}{5} = {r},
  cell{38}{6} = {r},
  cell{38}{7} = {r},
  cell{39}{2} = {r},
  cell{39}{3} = {r},
  cell{39}{4} = {r},
  cell{39}{5} = {r},
  cell{39}{6} = {r},
  cell{39}{7} = {r},
  cell{40}{2} = {r},
  cell{40}{3} = {r},
  cell{40}{4} = {r},
  cell{40}{5} = {r},
  cell{40}{6} = {r},
  cell{40}{7} = {r},
  cell{41}{2} = {r},
  cell{41}{3} = {r},
  cell{41}{4} = {r},
  cell{41}{5} = {r},
  cell{41}{6} = {r},
  cell{41}{7} = {r},
  cell{42}{2} = {r},
  cell{42}{3} = {r},
  cell{42}{4} = {r},
  cell{42}{5} = {r},
  cell{42}{6} = {r},
  cell{42}{7} = {r},
  cell{43}{2} = {r},
  cell{43}{3} = {r},
  cell{43}{4} = {r},
  cell{43}{5} = {r},
  cell{43}{6} = {r},
  cell{43}{7} = {r},
  cell{44}{2} = {r},
  cell{44}{3} = {r},
  cell{44}{4} = {r},
  cell{44}{5} = {r},
  cell{44}{6} = {r},
  cell{44}{7} = {r},
  cell{45}{2} = {r},
  cell{45}{3} = {r},
  cell{45}{4} = {r},
  cell{45}{5} = {r},
  cell{45}{6} = {r},
  cell{45}{7} = {r},
  cell{46}{2} = {r},
  cell{46}{3} = {r},
  cell{46}{4} = {r},
  cell{46}{5} = {r},
  cell{46}{6} = {r},
  cell{46}{7} = {r},
  cell{47}{2} = {r},
  cell{47}{3} = {r},
  cell{47}{4} = {r},
  cell{47}{5} = {r},
  cell{47}{6} = {r},
  cell{47}{7} = {r,fg=red},
  cell{48}{2} = {r},
  cell{48}{3} = {r},
  cell{48}{4} = {r},
  cell{48}{5} = {r},
  cell{48}{6} = {r},
  cell{48}{7} = {r},
  cell{49}{2} = {r},
  cell{49}{3} = {r},
  cell{49}{4} = {r},
  cell{49}{5} = {r},
  cell{49}{6} = {r},
  cell{49}{7} = {r},
  cell{50}{2} = {r},
  cell{50}{3} = {r},
  cell{50}{4} = {r},
  cell{50}{5} = {r},
  cell{50}{6} = {r},
  cell{50}{7} = {r},
  cell{51}{2} = {r},
  cell{51}{3} = {r},
  cell{51}{4} = {r},
  cell{51}{5} = {r},
  cell{51}{6} = {r},
  cell{51}{7} = {r},
  cell{52}{2} = {r},
  cell{52}{3} = {r},
  cell{52}{4} = {r},
  cell{52}{5} = {r},
  cell{52}{6} = {r},
  cell{52}{7} = {r},
  cell{53}{2} = {r},
  cell{53}{3} = {r},
  cell{53}{4} = {r},
  cell{53}{5} = {r},
  cell{53}{6} = {r},
  cell{53}{7} = {r},
  cell{54}{2} = {r},
  cell{54}{3} = {r},
  cell{54}{4} = {r},
  cell{54}{5} = {r},
  cell{54}{6} = {r},
  cell{54}{7} = {r},
  cell{55}{2} = {r},
  cell{55}{3} = {r},
  cell{55}{4} = {r},
  cell{55}{5} = {r},
  cell{55}{6} = {r},
  cell{55}{7} = {r},
  cell{56}{2} = {r},
  cell{56}{3} = {r},
  cell{56}{4} = {r},
  cell{56}{5} = {r},
  cell{56}{6} = {r},
  cell{56}{7} = {r,fg=red},
  cell{57}{2} = {r},
  cell{57}{3} = {r},
  cell{57}{4} = {r},
  cell{57}{5} = {r},
  cell{57}{6} = {r},
  cell{57}{7} = {r},
  cell{58}{2} = {r},
  cell{58}{3} = {r},
  cell{58}{4} = {r},
  cell{58}{5} = {r},
  cell{58}{6} = {r},
  cell{58}{7} = {r},
  cell{59}{2} = {r},
  cell{59}{3} = {r},
  cell{59}{4} = {r},
  cell{59}{5} = {r},
  cell{59}{6} = {r},
  cell{59}{7} = {r},
  cell{60}{2} = {r},
  cell{60}{3} = {r},
  cell{60}{4} = {r},
  cell{60}{5} = {r},
  cell{60}{6} = {r},
  cell{60}{7} = {r},
  cell{61}{2} = {r},
  cell{61}{3} = {r},
  cell{61}{4} = {r},
  cell{61}{5} = {r},
  cell{61}{6} = {r},
  cell{61}{7} = {r},
  cell{62}{2} = {r},
  cell{62}{3} = {r},
  cell{62}{4} = {r},
  cell{62}{5} = {r},
  cell{62}{6} = {r},
  cell{62}{7} = {r},
  cell{63}{2} = {r},
  cell{63}{3} = {r},
  cell{63}{4} = {r},
  cell{63}{5} = {r},
  cell{63}{6} = {r},
  cell{63}{7} = {r},
  cell{64}{2} = {r},
  cell{64}{3} = {r},
  cell{64}{4} = {r},
  cell{64}{5} = {r},
  cell{64}{6} = {r},
  cell{64}{7} = {r,fg=red},
  cell{65}{2} = {r},
  cell{65}{3} = {r},
  cell{65}{4} = {r},
  cell{65}{5} = {r},
  cell{65}{6} = {r},
  cell{65}{7} = {r,fg=red},
  cell{66}{2} = {r},
  cell{66}{3} = {r},
  cell{66}{4} = {r},
  cell{66}{5} = {r},
  cell{66}{6} = {r},
  cell{66}{7} = {r,fg=red},
  cell{67}{2} = {r},
  cell{67}{3} = {r},
  cell{67}{4} = {r},
  cell{67}{5} = {r},
  cell{67}{6} = {r},
  cell{67}{7} = {r},
  cell{68}{2} = {r},
  cell{68}{3} = {r},
  cell{68}{4} = {r},
  cell{68}{5} = {r},
  cell{68}{6} = {r},
  cell{68}{7} = {r},
  cell{69}{2} = {r},
  cell{69}{3} = {r},
  cell{69}{4} = {r},
  cell{69}{5} = {r},
  cell{69}{6} = {r},
  cell{69}{7} = {r},
  cell{70}{2} = {r},
  cell{70}{3} = {r},
  cell{70}{4} = {r},
  cell{70}{5} = {r},
  cell{70}{6} = {r},
  cell{70}{7} = {r},
  cell{71}{2} = {r},
  cell{71}{3} = {r},
  cell{71}{4} = {r},
  cell{71}{5} = {r},
  cell{71}{6} = {r},
  cell{71}{7} = {r},
  cell{72}{2} = {r},
  cell{72}{3} = {r},
  cell{72}{4} = {r},
  cell{72}{5} = {r},
  cell{72}{6} = {r},
  cell{72}{7} = {r},
  cell{73}{2} = {r},
  cell{73}{3} = {r},
  cell{73}{4} = {r},
  cell{73}{5} = {r},
  cell{73}{6} = {r},
  cell{73}{7} = {r,fg=red},
  cell{74}{2} = {r},
  cell{74}{3} = {r},
  cell{74}{4} = {r},
  cell{74}{5} = {r},
  cell{74}{6} = {r},
  cell{74}{7} = {r},
  cell{75}{2} = {r},
  cell{75}{3} = {r},
  cell{75}{4} = {r},
  cell{75}{5} = {r},
  cell{75}{6} = {r},
  cell{75}{7} = {r},
  cell{76}{2} = {r},
  cell{76}{3} = {r},
  cell{76}{4} = {r},
  cell{76}{5} = {r},
  cell{76}{6} = {r},
  cell{76}{7} = {r},
  cell{77}{2} = {r},
  cell{77}{3} = {r},
  cell{77}{4} = {r},
  cell{77}{5} = {r},
  cell{77}{6} = {r},
  cell{77}{7} = {r,fg=red},
  cell{78}{2} = {r},
  cell{78}{3} = {r},
  cell{78}{4} = {r},
  cell{78}{5} = {r},
  cell{78}{6} = {r},
  cell{78}{7} = {r},
  cell{79}{2} = {r},
  cell{79}{3} = {r},
  cell{79}{4} = {r},
  cell{79}{5} = {r},
  cell{79}{6} = {r},
  cell{79}{7} = {r},
  cell{80}{2} = {r},
  cell{80}{3} = {r},
  cell{80}{4} = {r},
  cell{80}{5} = {r},
  cell{80}{6} = {r},
  cell{80}{7} = {r},
  cell{81}{2} = {r},
  cell{81}{3} = {r},
  cell{81}{4} = {r},
  cell{81}{5} = {r},
  cell{81}{6} = {r},
  cell{81}{7} = {r},
  cell{82}{2} = {r},
  cell{82}{3} = {r},
  cell{82}{4} = {r},
  cell{82}{5} = {r},
  cell{82}{6} = {r},
  cell{82}{7} = {r},
  cell{83}{2} = {r},
  cell{83}{3} = {r},
  cell{83}{4} = {r},
  cell{83}{5} = {r},
  cell{83}{6} = {r},
  cell{83}{7} = {r,fg=red},
  cell{84}{2} = {r},
  cell{84}{3} = {r},
  cell{84}{4} = {r},
  cell{84}{5} = {r},
  cell{84}{6} = {r},
  cell{84}{7} = {r},
  cell{85}{2} = {r},
  cell{85}{3} = {r},
  cell{85}{4} = {r},
  cell{85}{5} = {r},
  cell{85}{6} = {r},
  cell{85}{7} = {r},
  cell{86}{2} = {r},
  cell{86}{3} = {r},
  cell{86}{4} = {r},
  cell{86}{5} = {r},
  cell{86}{6} = {r},
  cell{86}{7} = {r},
  cell{87}{2} = {r},
  cell{87}{3} = {r},
  cell{87}{4} = {r},
  cell{87}{5} = {r},
  cell{87}{6} = {r},
  cell{87}{7} = {r,fg=red},
  cell{88}{2} = {r},
  cell{88}{3} = {r},
  cell{88}{4} = {r},
  cell{88}{5} = {r},
  cell{88}{6} = {r},
  cell{88}{7} = {r},
  cell{89}{2} = {r},
  cell{89}{3} = {r},
  cell{89}{4} = {r},
  cell{89}{5} = {r},
  cell{89}{6} = {r},
  cell{89}{7} = {r},
  cell{90}{2} = {r},
  cell{90}{3} = {r},
  cell{90}{4} = {r},
  cell{90}{5} = {r},
  cell{90}{6} = {r},
  cell{90}{7} = {r},
  cell{91}{2} = {r},
  cell{91}{3} = {r},
  cell{91}{4} = {r},
  cell{91}{5} = {r},
  cell{91}{6} = {r},
  cell{91}{7} = {r},
  cell{92}{2} = {r},
  cell{92}{3} = {r},
  cell{92}{4} = {r},
  cell{92}{5} = {r},
  cell{92}{6} = {r},
  cell{92}{7} = {r},
  cell{93}{2} = {r},
  cell{93}{3} = {r},
  cell{93}{4} = {r},
  cell{93}{5} = {r},
  cell{93}{6} = {r},
  cell{93}{7} = {r,fg=red},
  cell{94}{2} = {r},
  cell{94}{3} = {r},
  cell{94}{4} = {r},
  cell{94}{5} = {r},
  cell{94}{6} = {r},
  cell{94}{7} = {r,fg=red},
  cell{95}{2} = {r},
  cell{95}{3} = {r},
  cell{95}{4} = {r},
  cell{95}{5} = {r},
  cell{95}{6} = {r},
  cell{95}{7} = {r},
  cell{96}{2} = {r},
  cell{96}{3} = {r},
  cell{96}{4} = {r},
  cell{96}{5} = {r},
  cell{96}{6} = {r},
  cell{96}{7} = {r},
  vline{2} = {-}{},
  hline{1-2,96-97} = {-}{},
}
language & Allosaurus & Wav2Vec2Phoneme & MultIPA & ZIPA & Allophant & HuPER         \\
abk      & 0.61      & 0.40            & 0.32    & 0.44 & 1.00      & 0.35          \\
ace      & 0.24      & 0.18            & 0.15    & 0.15 & 1.00      & 0.29          \\
ady      & 0.36      & 0.32            & 0.35    & 0.32 & 1.00      & \textbf{0.30} \\
aeb      & 0.32      & 0.17            & 0.17    & 0.25 & 1.00      & 0.17          \\
afn      & 0.20      & 0.10            & 0.11    & 0.21 & 1.00      & 0.14          \\
afr      & 0.23      & 0.11            & 0.14    & 0.16 & 1.00      & \textbf{0.10} \\
agx      & 0.25      & 0.20            & 0.17    & 0.16 & 1.00      & 0.30          \\
ajp      & 0.38      & 0.12            & 0.13    & 0.21 & 1.00      & \textbf{0.12} \\
aka      & 0.20      & 0.10            & 0.13    & 0.18 & 1.00      & 0.14          \\
apc      & 0.25      & 0.14            & 0.13    & 0.14 & 1.00      & 0.18          \\
ape      & 0.34      & 0.17            & 0.14    & 0.15 & 1.00      & 0.19          \\
apw      & 0.30      & 0.19            & 0.18    & 0.21 & 1.00      & \textbf{0.10} \\
asm      & 0.21      & 0.10            & 0.09    & 0.18 & 1.00      & 0.20          \\
azb      & 0.25      & 0.15            & 0.12    & 0.18 & 1.00      & 0.20          \\
bam      & 0.27      & 0.22            & 0.23    & 0.32 & 1.00      & 0.28          \\
bem      & 0.14      & 0.08            & 0.05    & 0.07 & 1.00      & 0.08          \\
ben      & 0.32      & 0.27            & 0.24    & 0.28 & 0.18      & 0.40          \\
bfd      & 0.35      & 0.17            & 0.17    & 0.21 & 1.00      & 0.27          \\
bfq      & 0.21      & 0.20            & 0.15    & 0.31 & 1.00      & 0.20          \\
bhk      & 0.28      & 0.14            & 0.12    & 0.13 & 1.00      & 0.14          \\
bin      & 0.19      & 0.08            & 0.10    & 0.17 & 1.00      & 0.12          \\
brv      & 0.36      & 0.28            & 0.29    & 0.26 & 1.00      & 0.27          \\
bsq      & 0.33      & 0.20            & 0.22    & 0.30 & 1.00      & \textbf{0.05} \\
bwr      & 0.24      & 0.10            & 0.15    & 0.14 & 1.00      & 0.10          \\
cbv      & 0.30      & 0.25            & 0.25    & 0.25 & 1.00      & 0.47          \\
ces      & 0.21      & 0.10            & 0.08    & 0.13 & 0.09      & 0.08          \\
cha      & 0.16      & 0.05            & 0.06    & 0.15 & 1.00      & 0.06          \\
cji      & 0.40      & 0.30            & 0.21    & 0.47 & 1.00      & 0.32          \\
col      & 0.38      & 0.27            & 0.30    & 0.31 & 1.00      & \textbf{0.25} \\
cpn      & 0.26      & 0.17            & 0.19    & 0.35 & 1.00      & 0.17          \\
dag      & 0.30      & 0.17            & 0.14    & 0.34 & 1.00      & 0.16          \\
dan      & 0.24      & 0.16            & 0.19    & 0.23 & 0.24      & 0.18          \\
deg      & 0.13      & 0.07            & 0.07    & 0.06 & 1.00      & 0.07          \\
dyo      & 0.17      & 0.07            & 0.12    & 0.16 & 1.00      & 0.13          \\
efi      & 0.22      & 0.09            & 0.10    & 0.20 & 1.00      & 0.20          \\
ell      & 0.18      & 0.07            & 0.05    & 0.08 & 0.13      & 0.13          \\
ema      & 0.26      & 0.10            & 0.21    & 0.17 & 1.00      & 0.24          \\
eus      & 0.40      & 0.07            & 0.08    & 0.09 & 0.05      & 0.08          \\
ewe      & 0.36      & 0.14            & 0.18    & 0.33 & 1.00      & 0.22          \\
ffm      & 0.22      & 0.12            & 0.12    & 0.16 & 1.00      & 0.17          \\
fin      & 0.21      & 0.13            & 0.11    & 0.15 & 0.14      & 0.12          \\
fub      & 0.19      & 0.06            & 0.09    & 0.16 & 1.00      & 0.16          \\
gaa      & 0.25      & 0.17            & 0.20    & 0.23 & 1.00      & 0.26          \\
gla      & 0.32      & 0.14            & 0.16    & 0.30 & 1.00      & 0.18          \\
guj      & 0.18      & 0.13            & 0.14    & 0.16 & 1.00      & 0.18          \\
gwx      & 0.54      & 0.21            & 0.24    & 0.24 & 1.00      & \textbf{0.16} \\
hak      & 0.30      & 0.17            & 0.17    & 0.23 & 1.00      & 0.18          \\
hau      & 0.12      & 0.11            & 0.07    & 0.15 & 1.00      & 0.16          \\
haw      & 0.21      & 0.12            & 0.14    & 0.14 & 1.00      & 0.15          \\
heb      & 0.26      & 0.11            & 0.16    & 0.15 & 1.00      & 0.21          \\
hil      & 0.21      & 0.11            & 0.12    & 0.15 & 1.00      & 0.12          \\
hin      & 0.27      & 0.12            & 0.07    & 0.17 & 0.09      & 0.15          \\
hni      & 0.37      & 0.19            & 0.22    & 0.27 & 1.00      & 0.57          \\
hrv      & 0.21      & 0.09            & 0.11    & 0.13 & 1.00      & 0.16          \\
hun      & 0.30      & 0.15            & 0.15    & 0.26 & 0.16      & \textbf{0.10} \\
hye      & 0.26      & 0.10            & 0.11    & 0.13 & 1.00      & 0.17          \\
ibb      & 0.20      & 0.12            & 0.15    & 0.22 & 1.00      & 0.23          \\
ibo      & 0.30      & 0.15            & 0.14    & 0.22 & 1.00      & 0.30          \\
idu      & 0.23      & 0.13            & 0.12    & 0.16 & 1.00      & 0.19          \\
ilo      & 0.17      & 0.10            & 0.09    & 0.13 & 1.00      & 0.12          \\
isl      & 0.25      & 0.15            & 0.12    & 0.16 & 1.00      & 0.15          \\
its      & 0.64      & 0.11            & 0.18    & 0.29 & 1.00      & 0.20          \\
kan      & 0.18      & 0.10            & 0.07    & 0.12 & 1.00      & \textbf{0.07} \\
kea      & 0.31      & 0.24            & 0.20    & 0.25 & 1.00      & \textbf{0.13} \\
khm      & 0.27      & 0.22            & 0.21    & 0.24 & 1.00      & \textbf{0.15} \\
klu      & 0.37      & 0.17            & 0.21    & 0.28 & 1.00      & 0.22          \\
knn      & 0.20      & 0.13            & 0.11    & 0.25 & 1.00      & 0.20          \\
kri      & 0.15      & 0.08            & 0.09    & 0.10 & 1.00      & 0.38          \\
kub      & 0.17      & 0.09            & 0.10    & 0.16 & 1.00      & 0.15          \\
kye      & 0.25      & 0.20            & 0.18    & 0.28 & 1.00      & 0.29          \\
lad      & 0.25      & 0.15            & 0.13    & 0.12 & 1.00      & 0.16          \\
lar      & 0.48      & 0.14            & 0.18    & 0.30 & 1.00      & \textbf{0.11} \\
lav      & 0.35      & 0.20            & 0.19    & 0.22 & 1.00      & 0.22          \\
led      & 0.42      & 0.23            & 0.15    & 0.21 & 1.00      & 0.25          \\
lgq      & 0.25      & 0.11            & 0.16    & 0.41 & 1.00      & 0.23          \\
lit      & 0.26      & 0.14            & 0.15    & 0.18 & 0.14      & \textbf{0.12} \\
lkt      & 0.26      & 0.18            & 0.17    & 0.17 & 1.00      & 0.21          \\
lug      & 0.16      & 0.10            & 0.15    & 0.17 & 1.00      & 0.16          \\
mak      & 0.29      & 0.17            & 0.19    & 0.18 & 1.00      & 0.25          \\
mal      & 0.37      & 0.19            & 0.18    & 0.22 & 1.00      & 0.46          \\
mlt      & 0.27      & 0.13            & 0.09    & 0.17 & 0.19      & 0.16          \\
mya      & 0.28      & 0.24            & 0.25    & 0.25 & 1.00      & \textbf{0.16} \\
njm      & 0.26      & 0.17            & 0.09    & 0.38 & 1.00      & 0.21          \\
nld      & 0.20      & 0.13            & 0.13    & 0.14 & 0.20      & 0.13          \\
ozm      & 0.22      & 0.14            & 0.16    & 0.18 & 1.00      & 0.19          \\
pam      & 0.24      & 0.16            & 0.19    & 0.18 & 1.00      & \textbf{0.08} \\
pes      & 0.23      & 0.14            & 0.15    & 0.17 & 1.00      & 0.19          \\
prs      & 0.27      & 0.16            & 0.17    & 0.20 & 0.99      & 0.17          \\
run      & 0.27      & 0.09            & 0.13    & 0.11 & 1.00      & 0.21          \\
sbc      & 0.32      & 0.22            & 0.12    & 0.18 & 1.00      & 0.18          \\
tsw      & 0.30      & 0.10            & 0.11    & 0.21 & 1.00      & 0.14          \\
tzm      & 0.33      & 0.16            & 0.15    & 0.21 & 1.00      & \textbf{0.14} \\
wuu      & 0.45      & 0.36            & 0.33    & 0.31 & 1.00      & \textbf{0.22} \\
yue      & 0.29      & 0.14            & 0.15    & 0.21 & 1.00      & 0.16          \\
Avg.     & 0.28      & 0.15            & 0.15    & 0.21 & 0.90      & 0.19          
\end{longtblr}

\begin{table}[H]
\centering
\small
\caption{Baseline model checkpoints}
\label{tab:model_checkpoints}
\begin{tblr}{
  vline{2} = {-}{},
  hline{1,9} = {-}{0.08em},
  hline{2} = {-}{},
}
Model           & Checkpoint                                                                    \\
Allosaurus      & https://github.com/xinjli/allosaurus                                          \\
Allophant       & https://github.com/kgnlp/allophant                                            \\
Wav2Vec2Phoneme & https://huggingface.co/facebook/wav2vec2-xlsr-53-espeak-cv-ft                 \\
MultIPA         & https://huggingface.co/ctaguchi/wav2vec2-large-xlsr-japlmthufielta-ipa1000-ns \\
ZIPA            & https://github.com/lingjzhu/zipa                                              \\
POWSM           & https://huggingface.co/espnet/powsm \\
Wav2Vec2-en & https://huggingface.co/Bluecast/wav2vec2-Phoneme
\end{tblr}
\end{table}

\section{Broader Impacts}

This work aims to advance speech and language technologies by improving the reliability and interpretability of phonetic representations. More accurate and acoustically grounded phonetic modeling can substantially benefit assistive applications in education, healthcare, and accessibility, where transcription errors and hallucinated outputs often undermine user trust and efficiency. By reducing such errors, HuPER has the potential to support more reliable screening, assessment, and communication tools for diverse populations. Beyond immediate applications, HuPER provides a scalable phonetic foundation that may facilitate the development of more expressive and flexible speech generation systems, including text-to-speech models based on acoustically meaningful tokens. More broadly, by framing speech perception as an adaptive and modular inference process, this work points toward future speech foundation systems that integrate perception, reasoning, and control, rather than isolated task-specific models.
\newpage

\newpage
\section*{NeurIPS Paper Checklist}

\begin{enumerate}

\item {\bf Claims}
    \item[] Question: Do the main claims made in the abstract and introduction accurately reflect the paper's contributions and scope?
    \item[] Answer: \answerYes{} 
    \item[] Justification: The abstract and introduction state the scope of HuPER as a human-inspired framework for adaptive phonetic perception over acoustic--phonetic evidence and linguistic knowledge. The stated contributions correspond to the proposed HuPER-Recognizer, HuPER-Perceiver, HuPER-Scheduler, and Dysfluent WFST based inference routes. The empirical claims are restricted to the evaluated English variation benchmarks, zero-shot transfer on VoxAngeles, and nfvPPA experiments under weak acoustic evidence or with available reference text.

\item {\bf Limitations}
    \item[] Question: Does the paper discuss the limitations of the work performed by the authors?
    \item[] Answer: \answerYes{} 
    \item[] Justification: Section Conclusion and Limitations discusses the main limitations of the current system, including the limited amount of human-verified spoken-phone supervision, the heuristic nature of the current Scheduler, and the fact that HuPER models only part of the inference pathways involved in human phonetic perception. The paper presents HuPER as a human-inspired computational framework rather than a complete cognitive or neural simulation.

\item {\bf Theory assumptions and proofs}
    \item[] Question: For each theoretical result, does the paper provide the full set of assumptions and a complete (and correct) proof?
    \item[] Answer: \answerYes{} 
    \item[] Justification: HuPER-Recognizer introduces the DRRC formulation, the missing-label view of self-training, and the consistency conditions used by the main theorem. Appendix B provides the formal setup, complete theorem statements, assumptions, and proofs. The theorem in the main text is presented as an informal summary of the full appendix result.

    \item {\bf Experimental result reproducibility}
    \item[] Question: Does the paper fully disclose all the information needed to reproduce the main experimental results of the paper to the extent that it affects the main claims and/or conclusions of the paper (regardless of whether the code and data are provided or not)?
    \item[] Answer: \answerYes{} 
    \item[] Justification: The paper describes the model backbone, training pipeline, proxy-label construction, phoneme-to-phone correction procedure, inference routes, evaluation metrics, datasets, baselines, and switching analysis. Algorithms 1--2, Appendix E, and the released code, labels, and model artifacts provide the information needed to reproduce or verify the main experimental results.

\item {\bf Open access to data and code}
    \item[] Question: Does the paper provide open access to the data and code, with sufficient instructions to faithfully reproduce the main experimental results, as described in supplemental material?
    \item[] Answer: \answerYes{} 
    \item[] Justification: The HuPER repository is publicly available and provides the core code for the phone recognizer and phoneme-to-phone Corrector, together with runnable notebooks and a minimal inference example for the Corrector. The repository documents the main code layout and usage entry points. Third-party raw datasets are referenced through their original sources and are not redistributed unless permitted by their access conditions. We will continue to maintain the released repository with the scripts, checkpoints, and instructions needed to reproduce the main evaluations.

\item {\bf Experimental setting/details}
    \item[] Question: Does the paper specify all the training and test details (e.g., data splits, hyperparameters, how they were chosen, type of optimizer) necessary to understand the results?
    \item[] Answer: \answerYes{} 
    \item[] Justification: Experimental Setup specifies the WavLM-Large backbone, CTC objective, TIMIT-to-LibriSpeech self-learning recipe, Corrector training configuration, optimizer, learning rates, batch size, freezing and unfreezing schedule, evaluation datasets, and PFER metric. Appendix E and the dataset and baseline appendices further document the evaluation protocol, label normalization, and benchmark construction details.

\item {\bf Experiment statistical significance}
    \item[] Question: Does the paper report error bars suitably and correctly defined or other appropriate information about the statistical significance of the experiments?
    \item[] Answer: \answerNo{} 
    \item[] Justification: The paper reports a statistical test for the distortion-based scheduling analysis, including the Spearman correlation between emission distortion and PFER and the corresponding $p$-value. It also reports per-dataset and per-language results under fixed evaluation splits and a unified evaluation pipeline. However, the main phone-recognition benchmark tables do not yet include run-level error bars, confidence intervals, or significance tests, because the principal comparisons are reported as deterministic evaluations of released checkpoints on fixed test sets.

\item {\bf Experiments compute resources}
    \item[] Question: For each experiment, does the paper provide sufficient information on the computer resources (type of compute workers, memory, time of execution) needed to reproduce the experiments?
    \item[] Answer: \answerYes{} 
    \item[] Justification: Experimental Setup reports the main compute used for the HuPER training pipeline, including 2$\times$A6000 GPUs for Corrector training and 8$\times$A6000 GPUs for HuPER-Recognizer training, together with the number of epochs, batch size, learning rates, and freezing schedule. The released instructions provide the remaining environment and execution details needed to reproduce the experiments.
    
\item {\bf Code of ethics}
    \item[] Question: Does the research conducted in the paper conform, in every respect, with the NeurIPS Code of Ethics \url{https://neurips.cc/public/EthicsGuidelines}?
    \item[] Answer: \answerYes{} 
    \item[] Justification: We have reviewed the NeurIPS Code of Ethics and believe the work conforms to it. The study uses existing speech datasets, released model assets, and de-identified or properly accessed evaluation resources under their intended research use. The paper reports limitations and potential risks, and does not involve deceptive data collection, unsafe deployment, or undisclosed intervention with human participants.

\item {\bf Broader impacts}
    \item[] Question: Does the paper discuss both potential positive societal impacts and negative societal impacts of the work performed?
    \item[] Answer: \answerYes{} 
    \item[] Justification: The Appendix Broader Impacts discusses both the potential benefits and risks of HuPER. HuPER may support robust speech processing, pronunciation variation analysis, low-resource phonetic modeling, and speech accessibility, but phonetic transcription systems could also be misused for unwanted speech analysis or for drawing unreliable conclusions about individual speakers or speaker groups. The paper therefore presents HuPER as a research system evaluated on controlled benchmarks, with explicit limitations rather than deployment-ready claims.
    
\item {\bf Safeguards}
    \item[] Question: Does the paper describe safeguards that have been put in place for responsible release of data or models that have a high risk for misuse (e.g., pre-trained language models, image generators, or scraped datasets)?
    \item[] Answer: \answerNA{} 
    \item[] Justification: The released HuPER artifacts consist of research code, phone-recognition components, Corrector code, and usage examples rather than generative models, large language models, surveillance systems, or scraped unsafe datasets. The release does not redistribute third-party raw speech datasets unless permitted by their original access conditions. The paper presents HuPER as a research system and discusses limitations and potential misuse in the Impact Statement.

\item {\bf Licenses for existing assets}
    \item[] Question: Are the creators or original owners of assets (e.g., code, data, models), used in the paper, properly credited and are the license and terms of use explicitly mentioned and properly respected?
    \item[] Answer: \answerYes{} 
    \item[] Justification: The paper cites the original sources for the datasets, pretrained models, G2P tools, baseline recognizers, and WFST-related components used in the experiments. Existing assets are used through their original access channels and under their intended research-use conditions. The release focuses on HuPER code, trained models, processed labels, and evaluation scripts, and does not redistribute third-party raw datasets unless redistribution is permitted by the corresponding license or data-use agreement.

\item {\bf New assets}
    \item[] Question: Are new assets introduced in the paper well documented and is the documentation provided alongside the assets?
    \item[] Answer: \answerYes{} 
    \item[] Justification: The paper introduces HuPER-related assets, including the released codebase, phoneme-to-phone Corrector, phone-recognition training code, and runnable examples. The repository documents the basic usage, minimal Corrector inference procedure, and code organization for the released HuPER components. The release distinguishes HuPER implementation artifacts from third-party datasets, pretrained models, and external tools, which are referenced through their original sources.

\item {\bf Crowdsourcing and research with human subjects}
    \item[] Question: For crowdsourcing experiments and research with human subjects, does the paper include the full text of instructions given to participants and screenshots, if applicable, as well as details about compensation (if any)? 
    \item[] Answer: \answerNA{} 
    \item[] Justification: The work does not conduct crowdsourcing experiments or recruit new human participants. The phonetic annotation and evaluation procedures are performed on existing speech data, and no participant-facing instructions, screenshots, or compensation procedures are involved in this study.

\item {\bf Institutional review board (IRB) approvals or equivalent for research with human subjects}
    \item[] Question: Does the paper describe potential risks incurred by study participants, whether such risks were disclosed to the subjects, and whether Institutional Review Board (IRB) approvals (or an equivalent approval/review based on the requirements of your country or institution) were obtained?
    \item[] Answer: \answerYes{} 
    \item[] Justification: The nfvPPA experiments use a private clinical speech dataset accessed under IRB-approved and data-use restricted conditions. The data are de-identified before analysis, used only for controlled research evaluation, and are not redistributed. We report only aggregate experimental results and omit utterance identifiers or audio references in failure-case analysis to avoid disclosure risk. No new participants were recruited and no new participant-facing intervention was conducted for this study.

\item {\bf Declaration of LLM usage}
    \item[] Question: Does the paper describe the usage of LLMs if it is an important, original, or non-standard component of the core methods in this research? Note that if the LLM is used only for writing, editing, or formatting purposes and does \emph{not} impact the core methodology, scientific rigor, or originality of the research, declaration is not required.
    \item[] Answer: \answerNA{} 
    \item[] Justification: LLMs are not used as an important, original, or non-standard component of HuPER's core methodology. The proposed framework relies on acoustic--phonetic recognition, corrected proxy-phone supervision, explicit constraint graphs, and distortion-controlled routing rather than LLM-based inference.

\end{enumerate}

\end{document}